\documentclass[a4paper,10pt]{article}
\usepackage{booktabs} % For formal tables
\usepackage[ruled, linesnumbered, noend]{algorithm2e} % For algorithms
\usepackage{natbib} % For author-year citations
\usepackage[margin=1in]{geometry}

\SetAlFnt{\small}
\SetAlCapFnt{\small}
\SetAlCapNameFnt{\small}
\SetAlCapHSkip{0pt}
\IncMargin{-\parindent}

\setcitestyle{authoryear}

\title{All finite lattices are stable matching lattices\footnote{A previous version of this paper was circulated under the title  \emph{Non-distributive Lattices, Stable Matchings, and Linear Optimization}. An extended abstract with that title appeared in the proceedings of IPCO 2025 \citep{en2025non}.}}

\usepackage{authblk}

\author{Christopher En\thanks{Corresponding Author: ce2456@columbia.edu}} 
\author{Yuri Faenza\thanks{yf2414@columbia.edu}}
\affil{Columbia University, IEOR Department}

\usepackage[utf8]{inputenc}
\usepackage{amsthm}
\usepackage{amsmath,amssymb}
\usepackage{float}

\usepackage{thmtools, thm-restate}

\usepackage{graphicx}
\usepackage{subcaption}
\usepackage{xcolor}

\usepackage{hyperref}
\hypersetup{colorlinks=true,linkcolor=blue,citecolor=blue,urlcolor=blue}

\usepackage{soul}

\usepackage{color-edits}
\addauthor{ce}{red}
\addauthor{yf}{purple}

\usepackage[normalem]{ulem}

\usepackage{enumerate}

\theoremstyle{plain}
\newtheorem{theorem}{Theorem}
\newtheorem{corollary}[theorem]{Corollary}
\newtheorem{lemma}[theorem]{Lemma}
\newtheorem{proposition}[theorem]{Proposition}

\theoremstyle{definition}
\newtheorem{definition}[theorem]{Definition}
\newtheorem{example}{Example}[section]

\theoremstyle{remark}

\renewcommand{\[}{\begin{equation}}
\renewcommand{\]}{\end{equation}}

\newcommand{\bea}{\begin{eqnarray}}
\newcommand{\eea}{\end{eqnarray}}

\newcommand{\id}{\mathrm{id}}

\newcommand{\R}{\mathbb{R}}

\begin{document}

\maketitle

\begin{abstract}
We show that all finite lattices, including non-distributive lattices, arise as stable matching lattices when all agents have path-independent choice functions. This result answers an open question of Blair~\cite{blair1988lattice}. In the process, we introduce new tools to reason on general lattices for optimization purposes: the \emph{partial representation} of a lattice, which partially extends Birkhoff's representation theorem to non-distributive lattices; the \emph{distributive closure} of a lattice, which gives such a partial representation; and \emph{join constraints}, which can be added to the distributive closure to obtain a representation for the original lattice. Then, we use these techniques to show that the minimum cost stable matching problem under the same standard assumptions on choice functions is NP-hard, by establishing a connection with antimatroid theory.
\end{abstract}

%\tableofcontents 

\section{Introduction}

Since their introduction by Gale and Shapley~\citet{gale1962college}, stable matchings have been studied extensively by the economics, computer science, and optimization communities. They are today applied in areas as different as school choice~\citet{abdulkadirouglu2003school}, online dating~\citet{hitsch2010matching} and ride-sharing~\citet{wang2018stable}, and have been the subject of extensive theoretical research, often collected in monographs or large portions of them~\citep{immorlica2023online,gusfield1989stable,manlove2013algorithmics,roth1992two}.

One reason for the popularity of stable matchings is their rich mathematical structure,  which has been leveraged to design efficient algorithms. In particular, in any instance of Gale and Shapley's classical one-to-one model (the \emph{marriage} setting), stable matchings can be organized to form a distributive lattice under a natural partial order~\citet{knuth1976marriages}. This fact has been used to devise an efficient combinatorial algorithms for various problems over the set of stable matchings (see, e.g.,~\citet{faenza2021stable,faenza2024two}) including the problem of finding a stable matching of minimum cost, when costs are defined over possible pairs of agents~\citet{irving1987efficient}. The latter algorithm has been extended to more general  models of many-to-many two-sided matchings~\citep{bansal2007polynomial,faenza2023affinely}, still exploiting the underlying distributive lattice structure. Moreover,~\citet{faenza2023affinely,kiraly2008total} showed that even the linear description of the stable matching polytope~\citep{rothblum1992characterization,vate1989linear} and its properties can be understood, or even deduced, starting from the distributive lattice structure of stable matchings.

Interestingly, \citet{blair1984every} (and later~\citet{gusfield1989stable} with a simpler construction) proved the converse implication of~\citet{knuth1976marriages}: every finite distributive lattice arises as the stable matching lattice of some one-to-one instance (all lattices, and more generally all sets, from this paper are assumed to be finite). Thus, the class of stable matchings in the one-to-one model is universal enough to model any distributive lattice structure, making the theory of stable matchings relevant beyond its original scope. 

Considerable attention has been devoted to matching markets more general than the ones investigated in~\citet{bansal2007polynomial,faenza2023affinely,gale1962college}. In particular, \citet{roth1984stability} proved that, in two-sided markets where agents' choice functions are consistent and substitutable (see Section~\ref{sec:preliminaries} for definitions), a stable matching exists and can be found efficiently. Interestingly, stable matchings in such models are also known to form a lattice, which may however be non-distributive~\citep{blair1988lattice}. This is a fundamental difference with respect to simpler models, since the absence of a distributive lattice structure prevents the application of Birkhoff's Representation Theorem, on which many of the algorithms, implicitly or explicitly, rely (see Section~\ref{sec:preliminaries} for definitions and~\citet{faenza2023affinely,gusfield1989stable} for discussions). As a consequence, beyond Roth's result mentioned above, to the best of our knowledge no algorithm is known for the matching markets studied in~\citet{roth1984stability}, despite their relevance for applications (see, e.g.,~\citet{echenique2006theory,konishi2006credible,sonmez2010course}). More generally, while~\citet{faenza2023affinely} explored the algorithmic implications of associating a distributive lattice structure to the feasible region of any discrete optimization problem (not necessarily a matching problem), to the best of our knowledge nothing is known on the role of (non-distributive) lattice structures for devising algorithms. This seems to be a relevant question, since general lattices appear much more frequently than distributive lattices, in settings such as matching with contracts \citep{hatfield2005matching}, formal concept analysis \citep{buelohlavek2004concept, garg2015introduction}, and symbolic computation \citep{reynolds1970transformational}, among others.

\smallskip 

\noindent {\bf Our contributions.} Our first main contribution is a universality theorem for stable matchings: Every finite lattice arises as the lattice of stable matchings for a \emph{standard matching market instance} - that is, a matching market instance with substitutable and consistent choice functions, given as its set of agents together with an algorithm that computes the choice functions of the agents in time polynomial in the number of agents. A formal definition of standard matching markets is given in Section~\ref{sec:preliminaries:problem_statement}. 

\begin{theorem}[All finite lattices are stable matching lattices]\label{thm:main1} Let ${\cal L}=(X,\succeq)$ be a finite lattice. Then, there is a standard matching market instance $I^*$ such that ${\cal L}$ is order-isomorphic to the lattice of stable matchings of $I^*$. Furthermore, $I^*$ has $O(|X|^4)$ agents, and can be constructed in $O(| X|^6)$ time.
\end{theorem}

Similarly to the results from~\citet{blair1984every,knuth1976marriages} for distributive stable matching lattices in the one-to-one model, Theorem~\ref{thm:main1} together with~\citet{blair1988lattice} establishes an equivalence between stable matching lattices and general lattices. It also answers an open question of \citet{blair1988lattice}, thus showing that standard matching markets are substantially more general than matching markets from~\citet{gale1962college,faenza2023affinely,bansal2007polynomial}. This suggests that positive algorithmic results from~\citet{gale1962college,faenza2023affinely,bansal2007polynomial} may not extend to this setting. Indeed, for our second main contribution, we show that the minimum cost stable matching problem in standard matching market instances is NP-Hard.

\begin{theorem}[Linear optimization over matching markets is NP-Hard]\label{thm:main2} The following minimum cost stable matching problem is NP-Hard: \\
Given: a standard matching market instance $I$ with agents $F \uplus W$ and a cost function $c:F\times W\rightarrow \mathbb{Z}$.\\
Find: a stable matching $\mu$ of $I$ that minimizes $c(\mu)$. 
\end{theorem}

While both our main results only pertain to stable matchings, we prove them by introducing new tools that apply to general lattices. These tools include the \emph{partial representation} of a lattice ${\cal L}$, which partially extends Birkhoff's representation theorem to non-distributive lattices; the \emph{distributive closure} of a lattice, which gives such a partial representation by defining a distributive lattice that ``contains'' ${\cal L}$; and \emph{join constraints}, which can be added to the distributive closure to obtain a representation for the original lattice ${\cal L}$. We further show that a join constraint can be enforced on a stable matching lattice via a \emph{join constraint augmentation}, a construction that involves adding auxiliary agents and editing certain choice functions. These tools and the idea of iteratively pruning the distributive closure via join constraints may be useful in dealing more generally with optimization problems over non-distributive lattices, which, as discussed above, arise in many settings beyond stable matchings.

\smallskip

\noindent {\bf The techniques: Theorem 1.} The known fact \citep{blair1984every,gusfield1989stable} that every finite distributive lattice $\mathcal L = (X, \succeq)$ is the lattice of stable matchings in some one-to-one matching market can be proved as follows. Birkhoff's representation theorem~\citep{birkhoff1937rings} implies that there is an order-isomorphism between ${\cal L}$ and the lattice $(\mathcal D(X_j),\supseteq)$ of the lower closed sets $\mathcal D(X_j)$ of an associated poset $(X_j, \succeq)$, called the \emph{representation poset}. Then, given any poset $(P,\succeq_P)$, an algorithm by \citet{gusfield1989stable} produces a one-to-one matching market instance $I_{(P, \succeq_P)}$ whose stable matching lattice is order-isomorphic to $(\mathcal D(P),\supseteq)$ (see Theorem~\ref{thm:gusfield_irving}). When applied to $(X_j, \succeq)$, the algorithm produces therefore a one-to-one matching market instance whose stable matching lattice is order-isomorphic to  $(\mathcal D(X_j),\supseteq)$, hence to $\mathcal L$, concluding the proof.

While Birkhoff's theorem fails for a finite non-distributive lattice ${\cal L} = (X, \succeq)$, a relaxed version holds. In particular, in Section~\ref{sec:join_constraints} we define an order-embedding $\psi:X\to \mathcal D(X_j)$ (see Lemma~\ref{lem:dist_closure_partial_rep} and Figure~\ref{fig:distributive_closure}). By definition of order-embedding, $\psi$ is injective. However, in general, it is not surjective: some elements of $\mathcal D(X_j)$ may not correspond to elements of $\mathcal L$. 
We call therefore $(X_j,\succeq)$ a \emph{partial representation of ${\cal L}$}, and call $\psi$ the \emph{partial representation function} (see Definition~\ref{def:partial_rep}). Then, the Gusfield-Irving construction produces a matching market $I_{(X_j, \succeq)}$  with stable matching lattice order-isomorphic to $(\mathcal D(X_j),\supseteq)$, which in turn contains the image of the original lattice $(X, \succeq)$ via $\psi$. Our goal is now to modify the matching market instance $I_{(X_j, \succeq)}$ to an instance $I^*$, so that the lattice of stable matchings of $I^*$ is order-isomorphic to ${\cal L}$. To do this, we show that it suffices to apply a polynomial number of \emph{join constraint augmentations} to $I_{(X_j, \succeq)}$. Roughly speaking, a join constraint augmentation edits the matching market instance so that certain matchings lose stability, without adding new stable matchings. In this way, we can ``prune'' the stable matching lattice until the only remaining stable matchings form a lattice order-isomorphic to the original lattice $(X, \succeq)$. Join constraint augmentations can be seen as an implementation for the stable matching lattice of the more general concept of \emph{join constraints} (see Definition~\ref{def:join_constraint}). 

To enforce join constraint augmentations, we introduce the concept of \emph{extendable matching markets}. These markets are edited versions of the original instance $I_{(X_j,\succeq)}$ whose set of stable matchings, when projected to the agents of $I_{(X_j,\succeq)}$, is contained in the set of stable matchings of $I_{(X_j,\succeq)}$. We then construct a sequence of extendable matching markets, where each successive market enforces an additional join constraint on the stable matchings of $I_{(X_j,\succeq)}$.

\smallskip    

\noindent {\bf The techniques: Theorem 2.} Our starting point is a connection between lattices and antimatroids, first explored in \citet{edelman1985theory, dilworth1990lattices, avann1961application}: one can univocally define an antimatroid $(V,{\cal G})$ via its \emph{path poset} $(Q, \supseteq)$, see Definition~\ref{def:path_endpoint}. With this encoding, the problem of finding an element of the antimatroid minimizing a linear function over the ground set $V$ given $(Q, \supseteq)$ is well-posed. \citet{merckx2019optimization} showed that this problem is NP-Hard. After presenting some building blocks in Section~\ref{sec:hardness}, in Section~\ref{sec:proof:thm:main2} we prove Theorem~\ref{thm:main2} by constructing, for any path poset $(Q, \supseteq)$ and function $c: V \rightarrow \mathbb{Z}$, a standard matching market instance $I$ and cost function $c'$ over pairs of agents, so that the lattice of stable matchings of $I$ is order-isomorphic to the corresponding antimatroid $(V,{\cal G})$, and the minimizers of $c$ over ${\cal G}$ biject with stable matchings of $I$ minimizing $c'$. Our construction relies on techniques developed in the first part of the paper.

\smallskip 

\noindent{\bf Additional related work.} The relationship between stable matching lattices and distributive lattices has been studied in \citet{alkan2001preferences,blair1988lattice,blair1984every,gusfield1989stable,knuth1976marriages}. \citet{birkhoff1937rings} first showed a connection between distributive lattices and the poset of join-irreducible elements. \citet{bansal2007polynomial,faenza2023affinely,irving1987efficient} leveraged this relationship to propose algorithms for finding stable matchings of minimum cost. Different extensions of stability in two-sided markets using concepts from matroid theory have been explored, e.g., in~\citet{DBLP:journals/mor/IwataY20,fleiner2001matroid,fleiner2016matroid}. \citet{fleiner2000stable} studies another generalization of the stable marriage problem involving stable common antichains, which also displays a lattice structure, and shows that the minimum cost stable common antichain problem is NP-complete.

\section{Preliminaries}\label{sec:preliminaries}
As previously mentioned, all sets from this paper are finite, so finiteness will not be explicitly stated henceforth. Throughout the paper,  for a real-valued function $f:\mathbb{R}\rightarrow \mathbb{R}$ and $S\subseteq \R$, we let $f(S)=\sum_{i \in S}f(i)$. Similarly, we extend functions $\sigma: S_1\to S_2$ between sets $S_1, S_2$ to the power sets $\sigma: 2^{S_1}\to 2^{S_2}$ so that for $T\subseteq S_1$, $\sigma(T) = \bigcup_{s\in T}\{\sigma(s)\}$. Given two disjoint sets $T_1, T_2$, we write the disjoint union as $T_1\uplus T_2$. We denote the symmetric difference operator by $\triangle$. Given $n \in \mathbb{N}$, we let $[n]=\{1,2,\dots,n\}$. We next introduce the basic structures and tools used throughout the paper.

%\smallskip 

\subsection{Lattices, matching markets, and stability}\label{sec:preliminaries:problem_statement}

A poset $\mathcal L = (X, \succeq)$ is a \emph{lattice} if any pair of elements $a,b\in X$ has a unique \emph{join}, or least upper bound, denoted $a\vee b$, and a unique \emph{meet}, or greatest lower bound, denoted $a\wedge b$. Let $(X,\succeq)$ be a lattice. It follows by induction that for any $Y\subseteq X$, there exists a unique upper bound $\vee Y$ (again called join) and lower bound $\wedge Y$ (again called meet), where $\vee \emptyset$ (resp., $\wedge \emptyset$) is the minimal (resp., maximal) element of ${\cal L}$. In particular, $X$ contains a \emph{minimum} element $\wedge X$ and a \emph{maximum} element $\vee X$. A lattice is \emph{distributive} if the join and meet operations distribute over each other: $a\vee(b \wedge c) = (a\vee b)\wedge (a\vee c)$ and $a\wedge(b \vee c) = (a\wedge b)\vee (a\wedge c)$ for all $a,b,c\in X$. An element $a\in X$ not equal to the minimum element is \emph{join-irreducible} if there does not exist $Y\subseteq X\setminus \{a\}$ such that $\vee Y = a$. The set of join-irreducible elements of $(X,\succeq)$ is denoted by $X_j$, omitting the dependency on $\succeq$ which will always be clear from the context. Note that $(X_j, \succeq)$ forms a poset, with the order induced by the original lattice order. 

Given $x, y\in X$ with $x\prec y$, if there does not exist $z\in X$ such that $x\prec z\prec y$, then $x$ is an \emph{immediate predecessor} of $y$, and $y$ is an \emph{immediate successor} of $x$.

Throughout the paper, we assume that a lattice $(X, \succeq)$ is given as the set $X$ of elements, an array $A^{\vee}\in X^{X\times X}$ such that $A^{\vee}_{x,y} = x\vee y$, and an array $A^{\wedge}\in X^{X\times X}$ such that $A^{\wedge}_{x,y} = x\wedge y$. Note that this input is polynomially equivalent to an array $A^{\succ}\in \{-1,0,1\}^{X\times X}$ where $A^{\succ}_{x,y} = 1$ if $x\succ y$, $A^{\succ}_{x,y} = -1$ if $x\prec y$, and $A^{\succ}_{x,y} = 0$ otherwise. This is because the join and meet can be computed via enumeration, given the partial order.

Next, we define the main objects of study of this paper, matching markets, and related notions. Fix two disjoint sets of \emph{agents}, called \emph{firms} $F$ and \emph{workers} $W$. Each firm $f$ has a choice function $\mathcal C_f:2^W\to 2^W$, such that $\mathcal C_f(S)\subseteq S$ for all $S\subseteq W$. Similarly, each worker $w$ has a choice function $\mathcal C_w:2^F\to 2^F$, such that $\mathcal C_w(S)\subseteq S$ for all $S\subseteq F$. Then, a matching market instance is given by $I = (F, W, (\mathcal C_f)_{f\in F}, (\mathcal C_w)_{w\in W})$.

Given a matching market $I = (F, W, (\mathcal C_f)_{f\in F}, (\mathcal C_w)_{w\in W})$, a \emph{matching} $\mu$ is a mapping from $F\cup W$ to $2^{F\cup W}$ such that for all $w\in W$ and all $f\in F$, we have $\mu(w)\subseteq F$, $\mu(f)\subseteq W$, and $f\in \mu(w)$ if and only if $w\in\mu(f)$. A matching can also be viewed as a set of firm-worker pairs, so we write $(f,w)\in\mu$, $f\in\mu(w)$, and $w\in\mu(f)$ interchangeably. A matching is \emph{individually rational} if for every agent $a$, we have $\mathcal C_a(\mu(a)) = \mu(a)$. A firm-worker pair $(f,w)\not\in \mu$ is called a \emph{blocking pair} if $f\in \mathcal C_w(\mu(w)\cup\{f\})$ and $w\in \mathcal C_f(\mu(f)\cup\{w\})$. The solution concept we are concerned with is \emph{stability}: a matching is \emph{stable} if it is individually rational and admits no blocking pairs. The set of stable matchings of $I$ is denoted $\mathcal S(I)$ and, when $I$ is clear from context,  simply by $\mathcal S$. 

Next, we define two common assumptions on choice functions, \emph{substitutability} and \emph{consistency}, which we assume throughout the paper.

\begin{definition}[Substitutability and Consistency]
    A choice function $\mathcal C$ is substitutable if for any set of partners $S$, we have that $a\in \mathcal C(S)$ implies that for all $T\subseteq S$, $a\in \mathcal C(T\cup\{a\})$. It is consistent if for any sets of partners $S,T$, we have that $\mathcal C(S)\subseteq T\subseteq S$ implies $\mathcal C(S) = \mathcal C(T)$. If $\mathcal C$ is both substitutable and consistent, we say it is \emph{path-independent}\footnote{\citet{plott1973path} originally defined path-independence as the condition $\mathcal C(S\cup T) = \mathcal C(\mathcal C(S)\cup \mathcal C(T))$ for all sets $S,T$. \citet{aizerman1981general} showed that this definition is equivalent to substitutability and consistency.}.
\hfill $\diamond$ \end{definition}

\citet{blair1988lattice} showed that under path-independence, the set of stable matchings is nonempty and forms a lattice, when ordered following firm preferences.

\begin{theorem}[Stable Matching Lattice \citep{blair1988lattice}]\label{thm:stable_matching_lattice}
    Let $I=(F, W, (\mathcal C_f)_{f\in F},(\mathcal C_w)_{w\in W})$ be a matching market instance, such that all choice functions satisfy path-independence. Let ${\cal S}$ be the set of stable matchings on $I$ and for $\mu, \mu' \in {\cal S}$, let  $\mu\succeq \mu'$ if and only if $\mathcal C_f(\mu(f)\cup\mu'(f)) = \mu(f)$ for all $f\in F$. Then, $({\cal S},\succeq)$ forms a lattice, called the \emph{Stable Matching Lattice}. In particular, there exist $\mu_F, \mu_W \in {\cal S}$, called the \emph{firm-optimal} and \emph{worker-optimal} stable matchings respectively, such that $\mu_F \succeq \mu \succeq \mu_W$ for all $\mu \in {\cal S}$. %Then, $(\mathcal S, \succeq)$ forms a lattice.
\end{theorem}

The number of choice functions that are path-independent is doubly exponential in the number of agents of the market~\citep{echenique2007counting}. Hence, we cannot associate to each of them an algorithm that can be encoded in space polynomial in the number of agents of the market and computes the output of the function on any input. To bypass this issue, previous work has assumed oracle access to choice functions, see, e.g.,~\citet{faenza2023affinely}. Here we introduce and consider a subset called \emph{standard} matching markets. These are matching markets $(F,W,({\cal C}_f)_{f \in F},(\mathcal C_w)_{w \in W})$ whose choice functions are path-independent, and which are given in input as an algorithm that can be encoded in space polynomial in $|F|+|W|$ and is guaranteed to compute, for each input, its output in time polynomial in $|F|+|W|$. Furthermore, if for all agents $a$ and all $T\subseteq F\cup W$, we have $|\mathcal C_a(T)|\le 1$, we call it a \emph{one-to-one} instance.

\subsection{Birkhoff's representation theorem and the poset of rotations}\label{sec:preliminaries:representations}

We next recall basic tools from order theory used throughout the paper.

\begin{definition}[Lower closed set]
    Given a poset $(X,\succeq)$, a \emph{lower closed set} of $X$ is any set $S\subseteq X$ such that $x\in S$ and $x'\preceq x$ implies $x'\in S$. The set of all lower closed sets of $X$ is denoted $\mathcal D(X)$. Note that $(\mathcal D(X), \supseteq)$ is a lattice.
\hfill $\diamond$ \end{definition}

Order-embeddings and order-isomorphisms are monotone functions between posets which preserve the partial order structure.

\begin{definition}[Order-embeddings and order-isomorphisms]
    Let $(X, \succeq_X)$ and $(Y, \succeq_Y)$ be posets. A function $\psi:X\to Y$ is an \emph{order-embedding} of $(X, \succeq_X)$ into $(Y, \succeq_Y)$, if for any $x,x'\in X$, we have $x\succeq_X x'\iff \psi(x)\succeq_Y \psi(x')$. Note that such a function must be injective, as $\psi(x) = \psi(x')$ implies $x\succeq_X x'$ and $x\preceq_X x'$. An order-embedding is an \emph{order-isomorphism} if it is also bijective.
\hfill $\diamond$ \end{definition}

A representation is an order-isomorphism between a distributive lattice and the lower closed sets of a partial order.

\begin{definition}[Representation]
    Let $\mathcal L = (X, \succeq)$ be a distributive lattice and $(B, \succeq_B)$ be a poset. $(B, \succeq_B)$ is a \emph{representation} of $\mathcal L$ if there exists an order-isomorphism $\psi:(X, \succeq)\to (\mathcal D(B), \supseteq)$. $\psi$ is called a \emph{representation function}. By definition, a representation function is invertible.
\hfill $\diamond$ \end{definition}

Next, we introduce Birkhoff's representation theorem, which allows us to represent a distributive lattice using the poset of join-irreducible elements.

\begin{theorem}[Birkhoff's Representation Theorem \citep{birkhoff1937rings}]\label{thm:birkhoff}
     Let $\mathcal L = (X, \succeq)$ be a distributive lattice. Then, $(X_j, \succeq)$ is a representation of $\mathcal L$, with (Birkhoff) representation function $\psi_j:X\to \mathcal D(X_j)$ given by $\psi_j(x) = \{x'\in X_j\mid x'\preceq x\}$.
\end{theorem}

When a matching market is one-to-one, the stable matching lattice is distributive~\citep{knuth1976marriages}, and a representation is given by its poset of rotations~\citep{irving1987efficient}. Rotations are minimal differences between stable matchings. %For future use, we define them here more generally for any stable matching lattice.

\begin{definition}[Rotations]
    Let $\mathcal L = (\mathcal S, \succeq)$ be a stable matching lattice of a one-to-one market. Given two matchings $\mu,\mu'\in\mathcal S$ such that $\mu$ is an immediate predecessor of $\mu'$ in $\mathcal L$, the corresponding \emph{rotation} $\rho$ is given by $(\rho^+, \rho^-)$ where $\rho^+ = \mu'\setminus\mu$ and $\rho^- = \mu\setminus\mu'$. The set of all rotations of $\mathcal S$ is denoted $\Pi(\mathcal S)$.
\hfill $\diamond$ \end{definition}

We note that different pairs of stable matchings in the same instance can generate the same rotations. Furthermore, in general $|\Pi(\mathcal S)|\in O(|F\cup W|^2)$.

\begin{theorem}[\citet{gusfield1989stable}\label{thm:rotation_representation}]
    Let $(F, W, (\mathcal C_f)_{f\in F},(\mathcal C_w)_{w\in W})$ be a one-to-one stable matching instance, and let $(\mathcal S, \succeq)$ be the stable matching lattice, with join-irreducible elements $\mathcal S_j$. Then, $(\mathcal S, \succeq)$ is distributive. Furthermore, there is a bijection $\phi_j: \Pi(\mathcal S)\to \mathcal S_j$, which induces a partial order $\succeq$ on $\Pi(\mathcal S)$, such that $(\Pi(\mathcal S), \succeq)$ is a representation for $\mathcal S$ with representation function $\psi_{\mathcal S}: \mathcal S\to \mathcal D(\Pi(\mathcal S))$ so that
    \[\psi_{\mathcal S}^{-1}(R) = \left(\triangle_{\rho\in R}(\rho^-\triangle \rho^+)\right)\triangle \mu_W = \mu_W\cup\left(\bigcup_{\rho\in R}\rho^+\right)\setminus\left(\bigcup_{\rho\in R}\rho^-\right) \notag.\]
    Hence, $\psi^{-1}_{\mathcal S}(R)$ is the stable matching obtained from $\mu_W$ by adding all pairs in $\cup_{\rho \in R} \, \rho^+$ and then removing all those in $\cup_{\rho \in R}\, \rho^-$. We refer to $\psi_{\mathcal S}$ as the \emph{rotation representation function}.
\end{theorem}

Note that we define the order $\succeq$ on ${\cal S}$ from the firm-perspective, while lower closed sets of rotations are defined so that the empty set corresponds to the worker-optimal stable matching $\mu_W$. 

In essence, in the one-to-one case each matching can be seen as a unique set of rotations applied to the worker-optimal matching. If $\rho\in \psi_{\mathcal S}(\mu)$, we say that $\rho$ \emph{occurs} in $\mu$. \citet{gusfield1989stable} also showed constructively that each partial order is order-isomorphic to the rotation poset of some stable matching instance.

\begin{theorem}[\citet{gusfield1989stable}\label{thm:gusfield_irving}]
    Let $(X, \succeq)$ be a poset. Then, there exists a stable matching instance $I_{(X, \succeq)}$ such that $(X, \succeq)$ is order-isomorphic to the poset of rotations $(\Pi(\mathcal S), \succeq)$ of the stable matching lattice $(\mathcal S, \succeq)$ for $I_{(X, \succeq)}$, with order-isomorphism given by $\phi: X\to \Pi(\mathcal S)$. $I_{(X, \succeq)}$ is called the \emph{GI instance} of $(X,\succeq)$. Given $(X,\succeq)$, we can compute $I_{(X, \succeq)}$, $\Pi(\mathcal S)$, and $\phi$ in $O(|X|^2)$ running time.
\end{theorem}

It follows immediately from Theorem~\ref{thm:rotation_representation} and Theorem~\ref{thm:gusfield_irving} that the original poset $(X, \succeq)$ is also a representation for the stable matching lattice $\mathcal S$ of the Gusfield and Irving matching market instance $I_{(X, \succeq)}$.

\begin{corollary}\label{cor:GI_repr}
    Let $(X, \succeq)$ be a poset, and let $(\mathcal S, \succeq)$ be the stable matching lattice for the instance $I_{(X, \succeq)}$. Then $(X, \succeq)$ is a representation for $(\mathcal S, \succeq)$, with representation function $\phi^{-1}\circ\psi_{\mathcal S}$.
\end{corollary}

\section{Representing non-distributive lattices}\label{sec:join_constraints}

In this section, we introduce new tools for representing general non-distributive lattices, and show how to implement them on stable matching lattices. We will then use these tools in the next section to prove Theorem~\ref{thm:main1}. 

\subsection{Partial representation and distributive closure}\label{sec:join_constraints:partial_rep}

To extend Birkhoff's theorem to any (possibly non-distributive) lattice ${\cal L}$, we introduce the notion of partial representation. Intuitively, we relax the conditions on $\psi$ from order-embeddings to certain order-isomorphisms, and allow some lower closed sets of the representation poset to not correspond to elements of ${\cal L}$.

\begin{definition}[Partial representation]\label{def:partial_rep}
	Let $\mathcal L = (X,\succeq)$ be a lattice. Let $(B,\succeq_B)$ be a poset. Then, $(B,\succeq_B)$ is a \emph{partial representation} of $\mathcal L$ if there exists an order-embedding $\psi$ of $\mathcal L$ into $(\mathcal D(B), \supseteq)$ such that, if $\overline x\in X$ (resp., $\underline x\in X$) is the maximal (resp., minimal) element of $\mathcal L$, then $\psi(\overline x) = B$ (resp., $\psi(\underline x) = \emptyset$).
%\begin{itemize}	\item If $\overline x\in X$ is the maximal element of $\mathcal L$, then $\psi(\overline x) = B$.	\item If $\underline x\in X$ is the minimal element of $\mathcal L$, then $\psi(\underline x) = \emptyset$.\end{itemize}
$\psi$ is called a \emph{partial representation function}.
\hfill $\diamond$ \end{definition}

Note that a representation function for a distributive lattice is also a partial representation function.     For a lattice $\mathcal L = (X,\succeq)$, we define its \emph{distributive closure} to be the lattice $(\mathcal D(X_j), \supseteq)$. $(\mathcal D(X_j), \supseteq)$ can be used to define a ``canonical'' partial representation, as discussed next.

\begin{lemma}[Partial representation via the distributive closure]\label{lem:dist_closure_partial_rep}
    Let $\mathcal L = (X,\succeq)$ be a lattice. Then, $(X_j,\succeq)$ is a partial representation of $\mathcal L$, with partial representation function $\psi_X: X\to \mathcal D(X_j)$ given by
    $$\psi_X(x) = \{x'\in X_j\mid x'\preceq x\}.$$ 
    We refer to $\psi_X$ as the \emph{canonical} partial representation function for $(X, \succeq)$.
\end{lemma}

\begin{proof}
    First recall that, for $x \in X$, we have $\bigvee \psi_X(x) = x$ (see, e.g., \cite[Theorem 9.4]{garg2015introduction}). Now let $x, x' \in X$. We need to show that $x\succeq x'\iff \psi_X(x)\supseteq \psi_X(x')$.  
         
    To show one direction, suppose  $x\succeq x'$. Then, any $y\in X_j$ with $y\preceq x'$ also satisfies $y\preceq x$. Thus, $$\psi_X(x)=\{y \in X_j : y \preceq x\}\supseteq \{y \in X_j : y \preceq x'\}= \psi_X(x').$$ For the converse direction, assume $\psi_X(x)\supseteq \psi_X(x')$. Then $x = \bigvee\psi_X(x)\succeq \bigvee \psi_X(x') =  x'$, as required. %So, $x\succeq x'$ if and only if $\psi_X(x)\supseteq \psi_X(x')$.

    Next, consider the maximal element $\overline x\in X$. Then $\overline x\succeq x'$ for all $x'\in X_j$, so $\psi_X(\overline x) = X_j$. Let $\underline x\in X$ be the minimal element. Then $\underline x \prec x'$ for all $x'\in X_j$, and so $\psi_X(\underline x) = \emptyset$. We conclude that $X_j$ is a partial representation for $\mathcal L$, with partial representation function $\psi_X$.
\end{proof}

\begin{figure}
    \centering
    \begin{subfigure}[t]{.3\linewidth}
    \centering\includegraphics[scale=0.5]{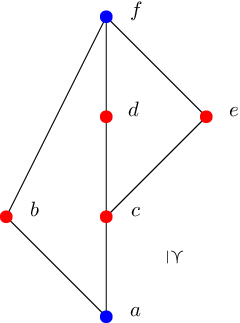}
    \caption{$(X, \succeq)$}
    \end{subfigure}
    \begin{subfigure}[t]{.3\linewidth}
    \centering\includegraphics[scale=0.5]{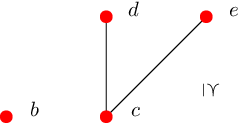}
    \caption{$(X_j, \succeq)$}
    \end{subfigure}
    \begin{subfigure}[t]{.3\linewidth}
    \centering\includegraphics[scale=0.5]{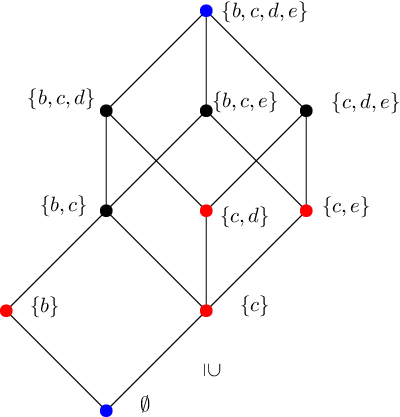}
    \caption{$(\mathcal D(X_j), \supseteq)$}
    \end{subfigure}
    \caption{Hasse diagrams of a non-distributive lattice $(X,\succeq)$, the corresponding poset of join-irreducible elements $(X_j,\succeq)$, and the distributive closure $(\mathcal D(X_j), \supseteq)$ of $(X,\succeq)$. See Example~\ref{ex:distributive_closure} and Example~\ref{example:join_constr_dist_closure} for discussions.}
    \label{fig:distributive_closure}
\end{figure}
We next present an example illustrating a non-distributive lattice and its distributive closure.

\begin{example}\label{ex:distributive_closure}
Consider the non-distributive lattice in Figure~\ref{fig:distributive_closure}(a). The join-irreducible elements are given by $X_j = \{b, c, d, e\}$, and the corresponding poset is shown in Figure~\ref{fig:distributive_closure}(b). The family of all lower closed sets, $\mathcal D(X_j)$, is shown in Figure~\ref{fig:distributive_closure}(c), along with the set containment relation that defines the distributive closure $(\mathcal D(X_j), \supseteq)$. The canonical partial representation function $\psi_X:X\to \mathcal D(X_j)$ is given by $\psi_X(x) = \{x'\in X_j\mid x'\preceq x\}$. In particular, $\psi_X$ is defined as
\begin{table}[H]
    \centering
    \begin{tabular}{|c||c|c|c|c|c|c|}
        \hline
        $x$ & $a$ & $b$ & $c$ & $d$ & $e$ & $f$ \\
        \hline
        $\psi_X(x)$ & $\emptyset$ & $\{b\}$ & $\{c\}$ & $\{c, d\}$ & $\{c, e\}$ & $\{b, c, d, e\}$  \\
        \hline
    \end{tabular}
\end{table}
We can check that $\psi_X$ preserves order, maps $\psi_X(a) = \emptyset$, and maps $\psi_X(f) = X_j$. It follows that $\psi_X$ is an order-embedding from $(X, \succeq)$ to $(\mathcal D(X_j), \supseteq)$. \hfill $\triangle$
\end{example}

Given a non-distributive lattice $(X, \succeq)$, we now have a framework for producing a stable matching instance $I$, such that there is an order-embedding from $(X, \succeq)$ to the stable matching lattice $(\mathcal S, \succeq)=(\mathcal S(I), \succeq)$. See Figure~\ref{fig:relations} for reference. We can start with the poset of join-irreducible elements $(X_j, \succeq)$, and construct the distributive closure $(\mathcal D(X_j), \supseteq)$. By Lemma~\ref{lem:dist_closure_partial_rep}, $(X_j,\succeq)$ partially represents $(X, \succeq)$ via the canonical partial representation function $\psi_X: X\to \mathcal D(X_j)$. Then, we can use Gusfield and Irving's construction (see Theorem~\ref{thm:gusfield_irving}) to produce a stable matching instance $I_{(X_j, \succeq)}$ with stable matching lattice $(\mathcal S, \succeq)$, such that the rotation poset $(\Pi(\mathcal S), \succeq)$ is order-isomorphic to $(X_j, \succeq)$ via $\phi: X_j\to \Pi(\mathcal S)$. The order-isomorphism $\phi$ then induces an order-isomorphism between the distributive closure $(\mathcal D(X_j), \supseteq)$ and the lower closed sets of the rotation poset $(\mathcal D(\Pi(\mathcal S)), \supseteq)$. Then, we know by Theorem \ref{thm:rotation_representation} that $(\Pi(\mathcal S), \succeq)$ represents $(\mathcal S, \succeq)$ with representation function $\psi_{\mathcal S}: \mathcal S\to \mathcal D(\Pi(\mathcal S))$. Finally, we deduced that $\psi_{\mathcal S}^{-1}\circ \phi\circ \psi_X$ is an order-embedding of $(X, \succeq)$ into the stable matching lattice $(\mathcal S, \succeq)$.

\begin{figure}[H]
    \centering
    \includegraphics[width=.7\linewidth]{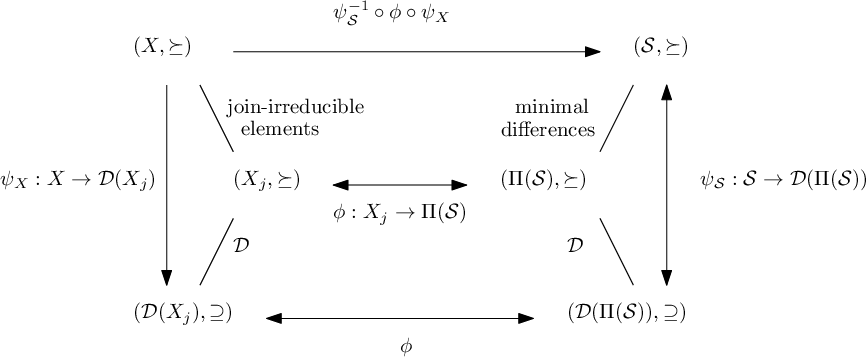}
    \caption{The functions connecting the non-distributive lattice $(X, \succeq)$, its distributive closure $(\mathcal D(X_j), \supseteq)$, the stable matching lattice $(\mathcal S, \succeq)$, and the rotation poset $(\Pi(\mathcal S), \succeq)$. The arrows indicate an order-embedding, and double arrows indicate an order-isomorphism.}
    \label{fig:relations}
\end{figure}

\subsection{Join constraints on lattices}\label{sec:join_constraints:join_constraints}

The main tool we use throughout the paper is the \emph{join constraint}, which is a constraint on a lattice that excludes certain elements.

\begin{definition}[Join constraint]\label{def:join_constraint}
    Let $(X, \succeq)$ be a lattice with partial representation $(B, \succeq)$ and partial representation function $\psi$. A \emph{join constraint} on $B$ is a pair $(\alpha, \beta)$ where $\alpha, \beta: 2^B\to \{0,1\}$ are functions of the form
    $$\alpha(T) = (\mathbf{1}_{b_{11}}(T)\vee \mathbf{1}_{b_{12}}(T)\vee\dots)\wedge (\mathbf{1}_{b_{21}}(T)\vee \mathbf{1}_{b_{22}}(T)\vee\dots)\wedge\dots, $$ $$\beta(T) = \mathbf{1}_{b_{1}}(T)\wedge\mathbf{1}_{b_{2}}(T)\wedge\dots$$
    for some $b_{11}, b_{12},\dots\in B$ such that $b_{i_1i_2}\not\prec b_{i_3i_4}$ for all $i_1,i_2,i_3,i_4$ and $b_1,b_2,\cdots\in B$, where $\mathbf{1}_b(T)$ is the indicator that $b\in T$. A set $T\in \mathcal D(B)$ \emph{satisfies} the join constraint $(\alpha, \beta)$ if, whenever $\alpha(T) = 1$, we also have $\beta(T) = 1$. Similarly, an element $x\in X$ satisfies the join constraint if $\psi(x)$ satisfies the join constraint. The \emph{arguments} of the join constraint are the elements $b_{11},b_{12},\dots$ and $b_{1}, b_2,\dots$. %We sometimes denote $(\alpha, \beta)$ by $(\{b_{11}, b_{12},\dots\}, \{b_1, b_2,\dots\})$.
\hfill $\diamond$ \end{definition}

Note that in the proof of Lemma~\ref{lem:poly_join_constrs} and Theorem~\ref{thm:main1}, only join constraints of a special form where $\alpha(T) = \mathbf{1}_{b_{11}}(T)\wedge \mathbf{1}_{b_{21}}(T)\wedge \dots$ are used. The more general form is required for the proof of Theorem~\ref{thm:main2}.

We next argue that, by applying a polynomial number of join constraints to the distributive closure, we can obtain the original lattice. Note that the running time of the algorithm from the lemma below does not depend on $|{\mathcal D}(X_j)|$, which may be exponential in $|X|$.

\begin{lemma}[Polynomially many join constraints suffice to describe a lattice]\label{lem:poly_join_constrs}
Let $(X,\succeq)$ be a lattice, with canonical partial representation $\psi_X$, and distributive closure $(\mathcal D(X_j), \supseteq)$. Then there exists a set $\Omega$ of $O(|X|^2)$ join constraints on $X_j$, such that $\psi_X$ is an order-isomorphism from $(X, \succeq)$ to\begin{equation}\label{eq:T-satisfies-Omega}(\{T\in \mathcal D(X_j)\mid T \text{ satisfies } (\alpha,\beta)\ \forall (\alpha, \beta)\in\Omega\}, \supseteq).\end{equation} Furthermore, the number of arguments $|(\alpha, \beta)|$ for each join constraint $(\alpha,\beta)\in \Omega$ is at most $O(|X|)$, and given $(X,\succeq)$, the set $\Omega$ can be computed in $O(|X|^3)$ time.
\end{lemma}

\begin{proof}
    For each $x, y\in X$, define $Y_{\alpha}^{xy} = \psi_{X}(x)\cup \psi_{X}(y)$ and $X_{\alpha}^{xy}=\{z \in Y^{xy}_\alpha \mid z \not\prec z'$ for any $z' \in Y^{xy}_\alpha\}$. That is, the set $X^{xy}_\alpha$ is the set of undominated elements of $Y_\alpha^{xy}$ with respect to $\succ$. Also define $X_{\beta}^{xy} = \psi_{X}(x\vee y)$. 
    
    Let $\alpha^{xy} = [T\mapsto \bigwedge_{z\in X_{\alpha}^{xy}}\mathbf{1}_{z}(T)]$  and $\beta^{xy} = [T\mapsto \bigwedge_{z\in X_{\beta}^{xy}}\mathbf{1}_{z}(T)]$. By construction, no two elements of $X_{\alpha}^{xy}$ dominate each other with respect to $\prec$, as required by Definition~\ref{def:join_constraint}.
    
    Consider the set of join constraints
    \[\Omega = \left\{(\alpha^{xy}, \beta^{xy}) \mid x, y\in X\right\}.\]
    $|\Omega|$ has size $O(|X|^2)$. Also, since $\alpha^{xy}$ and $\beta^{xy}$ each have $O(|X|)$ arguments, $\Omega$ can be computed in $O(|X|^3)$ time. We know that $\psi_X$ is an order-embedding. Thus, we simply need to show that its image is exactly $Q=(\{T\in \mathcal D(X_j)\mid T \text{ satisfies } (\alpha,\beta)\ \forall (\alpha, \beta)\in\Omega\}$. 

    To show one direction, fix any $w \in X$. We show $\psi_X(w)\in Q$. Let $x, y\in X$.  Assume $\alpha^{xy}(\psi_X(w)) = 1$. Then, $X_{\alpha}^{xy} \subseteq \psi_X(w)$. By definition, every element of $Y^{xy}_\alpha\setminus X_{\alpha}^{xy}$ is dominated by an element of $Y_{\alpha}^{xy}$, thus by transitivity of $\succ$ by an undominated element of $Y^{xy}_\alpha$. The undominated elements of $Y^{xy}_\alpha$ are exactly $X_\alpha^{xy}$. Since $\psi_X(w)$ is a lower closed set, we know therefore that $Y_{\alpha}^{xy}\subseteq \psi_X(w)$. By definition, $\psi_X(x)\cup \psi_X(y) =Y^{xy}_\alpha$. The latter two facts imply $w\succeq x,y$ by definition of order-embedding. Thus, $w\succeq z: = x\vee y$, and so $\psi_X(z)\subseteq\psi_X(w)$ as well, and so $\beta^{xy}(\psi_X(w)) = 1$. So, $w$ satisfies the join constraint $(\alpha^{xy}, \beta^{xy})$ for all $(\alpha^{xy}, \beta^{xy})\in \Omega$.

    To show the other direction, fix any $T\in \{T\in \mathcal D(X_j)\mid T \text{ satisfies } (\alpha,\beta)\ \forall (\alpha, \beta)\in\Omega\}$.  Define $w = \bigvee T$. We claim $\psi_X(w)=T$, which completes the proof. Clearly, $T \subseteq \{x' \in X_j : x' \preceq w \} = \psi_X(w)$. To show that $T = \psi_X(w)$, fix any ordering $x_1,\cdots, x_k$ over the elements in $T$, and define $y_i = x_1\vee\cdots\vee x_i$ for $i = 1,\cdots, k$. First observe $\psi_X(x_i)\subseteq T$ for each $1 \leq i \leq k$: indeed, since $T\in \mathcal D(X_j)$, we know that $v\in X_j$ with $v \preceq x_i$ satisfies $v\in T$, which implies that $\psi_X(x_i)\subseteq T$ by definition of $\psi_X$. 
    
    We next show by induction that $\psi_X(y_i)\subseteq T$ for each $1\le i\le k$. The base case $i=1$, follows from $y_1 = x_1$ and the previous discussion. Now assume that $\psi_X(y_i)\subseteq T$. Consider $(\alpha^{x_{i+1}y_i}, \beta^{x_{i+1}y_i})\in\Omega$. We know that $\psi_X(x_{i+1}), \psi_X(y_i)\subseteq T$, and so $\alpha^{x_{i+1}y_i}(T) = 1$. Since $T$ satisfies $(\alpha^{x_{i+1}y_i}, \beta^{x_{i+1}y_i})$, we know that $\beta^{x_{i+1}y_i}(T) = 1$. Thus, $\psi_X(y_{i+1}) = \psi_X(x_{i+1}\vee y_i)\subseteq T$, showing the inductive step. 
    
    We conclude that $\psi_X(w) = \psi_X(\bigvee T)= \psi_X(y_k) \subseteq T$. 
\end{proof}

Lemma~\ref{lem:poly_join_constrs} is illustrated in the next example.

\begin{example}\label{example:join_constr_dist_closure}
    Consider the lattice $(X, \succeq)$ in Figure~\ref{fig:distributive_closure}, along with its poset of join-irreducible elements $(X_j, \succeq)$ and distributive closure $(\mathcal D(X_j), \supseteq)$. Consider the set of join constraints $\Omega = \{(\alpha^1, \beta^1), (\alpha^2, \beta^2)\}$ given by
    $\alpha^1(T) = \mathbf{1}_{b}(T)\wedge \mathbf{1}_c(T), \beta^1(T) = \mathbf{1}_{d}(T)\wedge \mathbf{1}_e(T)$ and
    $\alpha^2(T) = \mathbf{1}_d(T) \wedge \mathbf{1}_e(T),
    \beta^2(T) = \mathbf{1}_{b}(T).$
    Then, the elements of $\mathcal D(X_j)$ that satisfy $(\alpha^1, \beta^1)$ and $(\alpha^2, \beta^2)$ are $\emptyset$,  $\{b\}$, $\{c\}$, $\{c, d\}$, $\{c, e\}$, and $\{b, c, d, e\}$. These are exactly the elements of $\psi_X(X)$. \hfill $\triangle$
\end{example}

\subsection{Implementing join constraints for the stable matching lattice}\label{sec:join_constraints:implement}

Following Lemma~\ref{lem:poly_join_constrs}, to prove Theorem~\ref{thm:main1}, we would like to create, given a lattice $(X,\succ)$, a standard matching market instance whose stable matchings admit an order-isomorphism with~\eqref{eq:T-satisfies-Omega}. In order to do that, we first introduce the following definition. Throughout the section, fix a GI instance  $I = (F, W, (\mathcal C_f)_{f\in F}, (\mathcal C_w)_{w\in W})$ with stable matching lattice $(\mathcal S, \succeq)$ and representation function $\psi_{\mathcal S}: \mathcal S\to \mathcal D(\Pi(\mathcal S))$. We also let $\Omega$ be a set of join constraints defined over $\Pi(\mathcal S)$.

%implement join constraints on a stable matching lattice, we need to ensure that the underlying matching market instance satisfies an appropriate invariant. For the remainder of this section, we fix a GI stable matching instance  $I = (F, W, \mathcal C_F, \mathcal C_W)$ with stable matching lattice $(\mathcal S, \succeq)$.

\begin{definition}[$\Omega$-Extension of a matching market]\label{def:extension-new}
    Let $I' = (F', W', $ $(\mathcal C'_f)_{f\in F'}, (\mathcal C'_w)_{w\in W'})$ be a matching market with stable matching lattice $(\mathcal S',\succeq')$ and such that $F\subseteq F'$, $W\subseteq W'$. We say that $I'$ is an \emph{$\Omega$-extension} of $I$ if there exists a function $\xi_{\mathcal S'}: \mathcal S'\to \mathcal S$ that preserves any edge in $2^{F\times W}$, and is an order-isomorphism from $(\mathcal S',\succeq')$ to
    $(\{\mu \in \mathcal S : \mu \hbox{ satisfies } (\alpha,\beta), \hbox{ for all $(\alpha,\beta)\in \Omega)$}\},\succeq)$.
    \hfill $\diamond$ 
\end{definition}

%With the terminology from Definition~\ref{def:extension-new}, our goal is to find an $\Omega$-extension

We next show that a standard $\Omega$-extension for a GI instance always exists.

\begin{theorem}[GI instances are $\Omega$-extendable] \label{thm:omega_extension}
There exists an $\Omega$-extension $I^*$ for $I$ with function $\xi_{\mathcal S^*}$ that has $O(|\Omega||W|)$ additional agents and is a standard matching market instance. Given $I$ as an input, $I^*$ can be computed in $O(|\Omega||F|(|W|+|(\alpha,\beta)^*|))$ time, where $|(\alpha,\beta)^*|$ is the maximum number of arguments in a join constraint in $\Omega$. Additionally, $\psi_{\mathcal S^*} = \psi_{\mathcal S}\circ \xi_{\mathcal S^*}: \mathcal S^* \rightarrow \mathcal D(\Pi(\mathcal S))$ is a partial representation function for the stable matching lattice $(\mathcal S^*, \succeq)$ of $I^*$.
\end{theorem}

Theorem~\ref{thm:omega_extension} is proved in the next subsection. Here, we give a roadmap. First, we introduce the concept of an \emph{extendable} matching market instance, which is a type of instance with a useful structure (Definition~\ref{def:extendable_instance}). Then, we introduce \emph{join constraint augmentation}, a construction that can be applied to an extendable matching market instance which then enforces a given join constraint on the stable matching lattice (Definition~\ref{def:join_constraint_augmentation}): We show that the new stable matching lattice after a join constraint augmentation corresponds exactly to the elements of the original lattice which satisfied the join constraint (Lemma~\ref{lem:join_constr_aug_respects_join_constr}). Then, we show that the join constraint augmentation preserves the extendability of the instance (Lemma~\ref{lem:join_constr_aug_preserves_partial_rep}), and that it can be implemented in polynomial time (Lemma~\ref{lem:jca_runtime}). Altogether, this allows us to show that join constraints can be iteratively enforced on a stable matching lattice using join constraint augmentations (Lemma~\ref{lem:jca_implemented}), which leads to Theorem~\ref{thm:omega_extension}.

\subsubsection{Proof of Theorem~\ref{thm:omega_extension}}\label{sec:proof:thm:omega_extension}

\paragraph{Alternative description of a matching market instance.} Consider a matching market $I=(F,W,(\mathcal C_f)_{f\in F}, (\mathcal C_w)_{w\in W})$. It will sometimes be useful to describe preferences of an agent $a$ with a \emph{preference relation} $\succeq_a$, which ranges over $2^W$ if $a\in F$ and over $2^F$ if $a\in W$, with corresponding preference list $P_a = (T_1,T_2,\cdots, T_n)$ where $T_k\succ_a T_{k+1}$ for all $k\in [n-1]$. We say that $\succ_a$ and the choice function $\mathcal C_a(T) = \max_{\succeq_a}2^T$ \emph{correspond} with each other, and that a set is \emph{acceptable} if and only if ${\cal C}_a(T) = T$, where $\max_{\succeq_a} 2^T$ selects the $\succeq_a$-maximal subset of $T$. Note that we often omit $\emptyset$ and any set ranked after $\emptyset$ from the preference list, as they do not affect the corresponding choice function. It is well-known that given any path-independent choice function $\mathcal C_a$ there exists a corresponding preference relation, though it may not be unique (see, e.g., \citet{yang2020rationalizable}). As such, the instance can also be written as $I = (F, W, (P_f)_{f\in F}, (P_w)_{w\in W})$. We use the latter description and the one using choice functions interchangeably. We abbreviate $P_a=(\{t_1\},\dots,\{t_n\})$ by $P_a=(t_1,\dots,t_n)$. Note that if $P_a=(t_1,\dots,t_n)$ for all agents $a$, then $I$ is a one-to-one instance.

Given an agent $a$, a set $T$, and an agent $b$ in the domain of $\mathcal C_a$, we say \emph{$a$ demands $b$ at $T$ (under ${\cal C}_a$)} if $b\in \mathcal C_a(T\cup\{b\})$. For agents $a,b$ from opposite sides of a matching market, we say that $b$ is \emph{acceptable} (for $a$) if $b \in {\cal C}_a(T)$ for some $T$. In the proofs, we make extensive use of the following observation: if $b$ is not acceptable for $a$, then $b \notin \mu(a)$ for any individually rational matching $\mu$.

\smallskip 

Throughout the rest of this section, fix a GI instance $I= (F, W, (\mathcal C_f)_{f\in F}, (\mathcal C_w)_{w\in W})=(F,W,(P_f)_{f\in F},$ $(P_w)_{w\in W})$ as in Section~\ref{sec:join_constraints:implement}, with stable matching lattice $({\cal S},\succeq)$ and representation function $\psi_{\mathcal S}: \mathcal S\to \mathcal D(\Pi(\mathcal S))$.

\paragraph{Choice functions used in the construction.} Recall that a standard matching market instance is given by its set $\overline F$ of firms, its set $\overline W$ of workers, together with an algorithm that can be encoded in space polynomial in $|\overline F|+|\overline W|$ and is guaranteed to compute the path-independent choice functions of any agent in time polynomial in $|\overline F|+|\overline W|$. In the instances constructed in this proof, we describe the choice function of each agent $a$ as one of the following:
\begin{enumerate}
    \item As a preference list $P'_a$, with length polynomial in the number of agents $|\overline F|+|\overline W|$. The resulting choice function is defined as the one corresponding to the preference order $\succ_a$ induced by $P'_a$, i.e., $\mathcal C_a(T)=\max_{\succ_a} 2^T$ for each input $T$.
    \item When $a=w$ for certain workers $w \in \overline W$, as an \emph{($\hat F$,$\hat f$, $\gamma$)-triggered} choice function, where $\hat F\subseteq \overline F$, $\hat f \in \overline F$,  and $\gamma$ is a function with binary output that takes as input a subset of $\overline F$ and runs in time polynomial in $|\overline F|+|\overline W|$. The corresponding choice function ${\mathcal C}'_w(T)$ for each set $T$ of firms in input is the outcome of the following process:
        \begin{enumerate}[i.]
            \item Select all firms in $T\cap \hat F$.
            \item If $\gamma(T) = 1$ and $\hat f\in T$, then additionally select $\hat f$. \end{enumerate}\label{def:join_constraint_augmentation:aux_worker_choice}
    \item When $a=f$ for certain firms $f \in \overline F$, as an \emph{if-else-$(\hat w,\hat W)$} choice function, where $\hat w \in \overline W$ and $\hat W \subseteq \overline W$. The corresponding choice function ${\mathcal C}'_f(T)$ for each set of workers $T$ in input is the outcome of the following process:
         \begin{enumerate}[i.]
         \item If $\hat w\in T$, select $\hat w$.
            \item Else, select $T\cap \hat W$.
        \end{enumerate}
    % \item When $a=f$ for $f \in F$ with $P_f=(w_1,\dots,w_k)$, as a \emph{$(\overline W_1,\dots,\overline W_k,\hat A)$-regular} choice function, where $\overline W_1,\dots,\overline W_k$ are disjoint subsets of $\overline W$, and $\hat A \subseteq (\overline W_1 \cup \overline W_2 \cup \dots \cup \overline W_k) \times \overline W$. The corresponding choice function $\mathcal C'_f(T)$ for each set of workers $T$ in input is the outcome of the following process:
    % \label{def:extendable_instance:reg_firms}
    %         \begin{enumerate}[i.]
    %             \item Let $i$ be the smallest index such that $T\cap \overline W_i\ne\emptyset$, if such an index exists. Select all workers in $T\cap \overline W_i$.
    %             \item Then, for each pair $(w_\ell, w')\in \hat A$,  such that $w' \in T$ and either $w_i \preceq_f w_\ell$ or the index $i$ as in part i.~does not exist, additionally select $w'$.
    %         \end{enumerate}

    \item When $a=f$ for $f \in F$ with $P_f=(w_1,\dots,w_k)$, as a \emph{$(\overline W_1,\dots,\overline W_k,\hat A)$-regular} choice function, where $\overline W_1,\dots,\overline W_k$ are disjoint subsets of $\overline W$, and $\hat A \subseteq \{w_1, \dots, w_k\} \times \overline W$.
    The corresponding choice function $\mathcal C'_f(T)$ for each set of workers $T$ in input is the outcome of the following process:
    \label{def:extendable_instance:reg_firms}
            \begin{enumerate}[i.]
                \item Let $i$ be the smallest index such that $T\cap \overline W_i\ne\emptyset$, if such an index exists.
                Select all workers in $T\cap \overline W_i$.
                \item Then, for each pair $(w_\ell, w')\in \hat A$,  such that $w' \in T$ and either $w_i \preceq_f w_\ell$ or the index $i$ as in part i.~does not exist, additionally select $w'$.
            \end{enumerate}

\end{enumerate}

In Section~\ref{sec:helpers}, we show in Lemma~\ref{lem:join_constr_aug_subs_cons} that, when appropriate parameters are selected, each of the above choice functions satisfies path-independence.

\paragraph{Extendable matching market instances and Join constraint augmentations.}

\begin{definition}[Extendable matching market instance]\label{def:extendable_instance}
    An instance $I' = (F', W', (\mathcal C'_f)_{f\in F'},$ $ (\mathcal C'_w)_{w\in W'}) = (F', W', (P'_f)_{f\in F'}, (P'_w)_{w\in W'})$ with stable matching lattice $(\mathcal S', \succeq')$ is an \emph{extendable} instance of $I$ if it satisfies the following properties.
    \begin{enumerate}
        \item \textbf{Substitutability and consistency}: All agents in $I'$ have substitutable and consistent choice functions. \label{def:extendable_instance:subs_cons}
    
        \item \textbf{Partition over agents}: $F'$ can be partitioned into $F\uplus F'_{aux}$ and $W'$ can be partitioned into $W'_{reg}\uplus W'_{aux}$. $W'_{reg}$ is called the set of \emph{regular}  workers, while $W'_{aux}$ (resp., $F'_{aux}$) is the set of \emph{auxiliary} workers (resp., firms). Moreover, $W'_{reg}$ can be partitioned into sets $\biguplus_{w\in W}copy'(w)$ and such that, for $w\in W$, we have $w \in copy'(w)$. We refer to elements of $copy'(w)$ as to \emph{copies of $w$.} \label{def:extendable_instance:partition}

        \item \textbf{Preference lists of workers}: 
            \begin{enumerate}
                \item For every worker $w\in W$, let $P_w = (f_1, f_2,\cdots, f_k) \subseteq F$. Then, we have $P_w' = P_w$ and, for every worker $w'\in copy'(w)\setminus\{w\}$, we have $P'_{w'} = (f_0, f_\ell, f_{\ell+1}, \cdots, f_k)$, for some $f_0\in F'_{aux}$  and $\ell \in [k]$ ($f_0$ and $\ell$ may be different for different workers $w'$). Note that $P'_{w'}$ is a one-to-one preference list.
                \item For every auxiliary worker $w \in W'_{aux}$, $P'_w$ may be arbitrary. 
            \end{enumerate}\label{def:extendable_instance:reg_workers}
        
        \item\label{it:I'-choice-firms} \textbf{Choice functions of firms}: 
        \begin{enumerate}
            \item For every firm $f\in F$, with $P_f = (w_1,w_2,\cdots, w_k)$, ${\mathcal C}'_f$ is a $(copy'(w_1),\dots, copy'(w_k)$,$A'_f$)-regular choice function, where $A'_f$ is a (possibly empty) set of pairs of workers $(w_\ell, w_{aux})$ where $\ell \in [k]$ and $w_{aux}\in W'_{aux}$, and such that each $w_{aux}\in W'_{aux}$ appears at most once in a pair in $A'_f$.

            \item For every auxiliary firm in $ f \in F'_{aux}$, $P'_f$ can be arbitrary. 
        \end{enumerate}

        \item \textbf{Order-embedding of stable matching lattice}: Let $\xi_{\cal S'}: F\cup W'_{reg}\to F\cup W$ be given by $\xi_{\cal S'}(f) = f$ for all $f\in F$, and $\xi_{\cal S'}(w') = w$ for all $w\in W$ and $w'\in copy'(w)$. Then, for every $\mu\in \mathcal S'$, $$\xi_{\cal S'}(\mu)=\{(\xi_{\cal S'}(f), \xi_{\cal S'}(w'))\mid (f, w')\in\mu,\ f \in F, w' \in W'_{reg}\}$$ is a stable matching in $\mathcal S$. Furthermore, $\xi_{\cal S'}$ is an order-embedding of $(\mathcal S', \succeq')$ into $(\mathcal S, \succeq)$, and the maximal and minimal elements of $\mathcal S$ are in the image of $\xi_{\cal S'}$. \label{def:extendable_instance:projection} \hfill $\diamond$ 
    \end{enumerate}
\end{definition}

We observe the following fact immediately from the definition of extendable instance.
\begin{corollary}
The GI instance is a standard extendable instance of itself.
\end{corollary}

For a join constraint $(\alpha,\beta)$ on $\Pi(\mathcal S)$,  we denote by $\Pi_{\alpha}$ the set of rotations that appear in the argument of $\alpha$, and by $\Pi_{\beta}$ the set of rotations that appear in the argument of $\beta$. For each $\rho\in \Pi_\alpha$, let $$F_\rho = \{f\mid (f, w)\in \rho^-\}$$ and let $F_\alpha = \bigcup_{\rho\in \Pi_{\alpha}} F_\rho$. For each $\rho\in \Pi_\beta$, let $$W_\rho = \{w\mid (f, w)\in \rho^+\}$$ and let $W_\beta = \bigcup_{\rho\in \Pi_\beta} W_\rho$. Note that the elements of $\Pi_\alpha\cup\Pi_\beta$ are rotations from the original rotation poset $\Pi = \Pi(\mathcal S)$, so $F_\alpha, W_\beta$ only contain agents from the original instance $I$.

Now, we define the \emph{join constraint augmentation}, which is the construction that allows us to implement a join constraint on an extendable instance. An example of such construction is given in Example~\ref{ex:join_constr_aug} in Appendix~\ref{app:join_constraints}. 

From now on, fix an extendable instance $I' = (F', W', (\mathcal C'_f)_{f\in F'}, (\mathcal C'_w)_{w\in W'}) = (F', W', (P'_f)_{f\in F'}, (P'_w)_{w\in W'})$ of $I$ with stable matching lattice $(\mathcal S',\succeq)$ and suppose we have a partial representation function $\psi_{\mathcal S'} = \psi_{\mathcal S}\circ \xi_{\mathcal S'}: {\mathcal S}' \rightarrow \mathcal D(\Pi(\mathcal S))$. In order for the definition to be well-posed, we rely on the following statement, which claims that if two rotations of a GI instance\footnote{The statement holds for any one-to-one matching market.} involve the same agent, then one must be a predecessor of the other.  

\begin{proposition}[\citet{gusfield1989stable}]\label{prop:GI_rotations_overlap}
    For any $\rho_1, \rho_2\in \Pi(\mathcal S)$, if there exists $f \in F$ such that (i) $(f, w_1)\in \rho_1^-$ and $(f, w_2)\in \rho_2^-$ or (ii) $(f, w_1)\in \rho_1^+$ and $(f, w_2)\in \rho_2^+$ for some $w_1, w_2 \in W$, then $\rho_1\succ \rho_2$ or $\rho_1\prec\rho_2$. Symmetrically, if there exists $w \in W$ such that (i) $(f_1, w)\in \rho_1^-$ and $(f_2, w)\in \rho_2^-$ or (ii) $(f_1, w)\in \rho_1^+$ and $(f_2, w)\in \rho_2^+$ for some $f_1, f_2 \in F$, then $\rho_1\succ \rho_2$ or $\rho_1\prec\rho_2$.
\end{proposition}

Recall that by definition of join constraint, $\Pi_\alpha$ is an antichain of the rotation poset. Thus, by Proposition~\ref{prop:GI_rotations_overlap}, the following holds. 
\begin{lemma}\label{obs:F-rho-do-not-intersect} For any two distinct rotations $\rho, \rho' \in \Pi_\alpha$, we have $F_\rho \cap F_{\rho'}=\emptyset$ and $W_\rho \cap W_{\rho'}=\emptyset$ %(COMMENT: Not sure we need the latter -- if not, the ``symmetric'' direction in the previous proposition can be omitted).
\end{lemma}

\begin{definition}[Join constraint augmentation]\label{def:join_constraint_augmentation}
       Let $(\alpha, \beta)$ be a join constraint on $\Pi(\mathcal S)$. Let $\pi: 2^F \to 2^\Pi$ be given by $\pi(T) = \{\rho\in \Pi_\alpha \mid F_\rho\subseteq T\}$. Then, the join constraint augmentation for $(\alpha, \beta)$ on input $I'$ outputs a new instance $I''=(F'',W'',(\mathcal C''_f)_{f\in F''},(\mathcal C''_w)_{w\in W''})=(F'',W'',(P''_f)_{f\in F''},(P''_w)_{w\in W''})$ given by the following procedure.
    \begin{enumerate}
        \item \textbf{Additional auxiliary agents}: Let $W''_{aux}$ be obtained by adding an auxiliary worker $w_0$ to $W'_{aux}$ and $F''_{aux}$ be obtained by adding an auxiliary firm $f_0$ to $F'_{aux}$. \label{def:join_constraint_augmentation:aux}
        
        \item \textbf{Additional regular workers}: For each $w_j\in W_{\beta}$, create a regular worker $w_j''$ and set $copy''(w_j)\gets copy'(w_j)\cup\{w_j''\}$. Define $W_\beta'' = \{w_j''\mid w_j\in W_{\beta}\}$. Thus, the new set of regular workers $W''_{reg}$ is given by $W'_{reg} \cup W''_\beta$. \label{def:join_constraint_augmentation:reg_workers}

        \item \textbf{Choice function of new auxiliary worker}: Let the choice function $\mathcal C''_{w_0}(T)$ for $w_0$ be  the $(F_\alpha$,$f_0$, $\gamma$)-triggered choice function, where $\gamma(T)=\alpha(\pi(F_\alpha\setminus T))$.
        \label{def:join_constraint_augmentation:add_workers}

        \item \textbf{Choice function of new auxiliary firm}: Let the choice function $\mathcal C''_{f_0}(T)$ for $f_0$ be the if-else-$(w_0, {W''_\beta})$ choice function. \label{def:join_constraint_augmentation:aux_firm_choice} 

        \item \textbf{Choice function of new regular workers}: For each $w_j''\in W_{\beta}''$, let $f_j$ be the least preferred firm in $P_{w_j}$ such that $(f_j, w_j)\in\rho^+$ for some $\rho\in \Pi_\beta$. Such a firm exists by definition of $W_\beta$. Then, let the preference list $P''_{{w_j''}}$ be given by:
        \begin{itemize}
            \item First, the set $\{f_0\}$.
            \item The same as the original preference list $P_{w_j}$ for $w_j$, starting from $\{f_j\}$ and continuing to the end of $P_{w_j}$.
        \end{itemize}\label{def:join_constraint_augmentation:reg_worker_choice}

        \item \textbf{Change in preferences of existing workers}: No change. 
        \item \textbf{Change in preferences of existing auxiliary firms}: No change.
        \item \textbf{Change in preferences of existing non-auxiliary firms}: For each firm $f \in F$, recall that $\mathcal C'_f$ is a $(copy'(w_1),\dots, copy'(w_{k}),A'_f)-$regular choice function in $I'$.
        \begin{itemize}\item For each firm $f\in F_{\alpha}$, let $w\in P_f$ be the worker such that $(f, w)\in \rho^-$ for some $\rho\in \Pi_\alpha$. Such a worker exists by definition of $F_\alpha$, and is unique by Proposition~\ref{prop:GI_rotations_overlap}. Set $A''_f = A'_f\cup\{(w, w_0)\}$.  \item For firms $f\in F\setminus F_{\alpha}$, set $A''_f = A'_f$. \item Then, for each firm $f\in F$, define $\mathcal C''_f$ to be the $(copy''(w_1),\dots, copy''(w_{k}),A_f'')$-regular choice function.\hfill $\diamond$\end{itemize}\label{def:join_constraint_augmentation:reg_firm_choice} 
    \end{enumerate} 
\end{definition}

\paragraph{An intuitive explanation of the mechanics of Join constraint augmentations.} As we will formally show next, the Join constraint augmentation is designed to create an extendable matching market instance $I''$ that is an $\{(\alpha,\beta)\}-$extension of $I'$. In particular, any stable matching $\mu''$ of $I''$ projects to a stable matching $\mu$ of $I$ such that  the logical implication $\alpha(\psi_{\mathcal S}(\mu))=1\implies \beta(\psi_{\mathcal S}(\mu))$ holds. That is, under the bijection $\psi_S$ mapping stable matchings of $I$ to closed sets of rotations, $\psi_S(\mu)$ contain all rotations in $\Pi_\alpha$ only if $\psi_S(\mu)$ contains all rotations in $\Pi_\beta$. Intuitively, this is achieved by creating a cascade of forced matches using the newly introduced auxiliary agents and the modification of choice functions of agents already in $I'$.
\begin{enumerate}
    \item If $\alpha$ evaluates to true, it means that
    $\mu=\left(\triangle_{\rho\in R}(\rho^-\triangle \rho^+)\right)\triangle \mu_W$ for some set $R$ of rotations of $I$ containing $\Pi_\alpha$. This means that each firm $f\in F_\alpha$ is matched to a worker that is strictly better than $w_f$ in $P_f$, where $w_f$ is such that $(f,w_f) \in \rho^-$ for some $(\rho^-,\rho^+)=\rho \in \Pi_\alpha$.
    Thus, by definition of regular choice function, $f$ does not demand $w_0$ at the current matching $\mu''$. In particular, $(f,w_0)\notin \mu''$. 
    \item Because of the previous bullet point and of the $(F_\alpha, f_0,\gamma)$-triggered choice function of $w_0$, $w_0$ demands $f_0$ at $\mu''$. Because of the if-else choice function of $f_0$, $f_0$ demands $w_0$ at $\mu''$. Thus, $(f_0,w_0) \in \mu''$.
    \item Since $(f_0, w_0) \in \mu''$, $f_0$ does not demand any copied regular workers in $W''_\beta$, who all had $f_0$ at the top of their preference lists.
    \item The new regular workers $w''_j\in W''_\beta$, having failed to match with their top choice $f_0$, must fall back to the remaining options on their preference lists. By design, these lists only including firms at or below the threshold $f_j$. Because non-auxiliary firms have regular choice functions, whenever they demand $w_j''$, they also demand $w_j$. Therefore, for the overall matching to remain stable, for each $w\in W$, all copies in $copy''(w)$ must have the same matches. This ensures all required rotations in $\Pi_\beta$ occur in $\mu$.
\end{enumerate}

\paragraph{Properties of Join constraint augmentations.} For the remainder of this section, fix a join constraint $(\alpha, \beta)$, and let $I'' = (F'', W'', (\mathcal C''_f)_{f\in F''}, (\mathcal C''_w)_{w\in W''})=(F'',W'',(P''_f)_{f\in F''},(P''_w)_{w\in W''})$ be the instance after applying the join constraint augmentation for $(\alpha, \beta)$ to $I'$. Recall that we assumed that $I'$ is an extendable instance of $I$. Let $\mathcal S''$ be the set of stable matchings of $I''$. We start by defining a function from $\mathcal S''$ to $\mathcal S$.

\begin{definition}
    For an agent $a\in F'\cup W'\cup W''_{\beta}$, define %the \emph{projection} 
    $\zeta_{\mathcal S''}(a)$ by 
    $$\zeta_{\mathcal S''}(a) = \begin{cases}
        a & \hbox{if } a\in F'\cup W'; \\
        w_j & \hbox{else if } a = w_j''\in copy''(w_j)\cap W''_\beta.  \\
    \end{cases}$$
    Then, for any $\mu''\in \mathcal S''$, define $\zeta_{\mathcal S''}(\mu'')$ by
    $$\{(\zeta_{\mathcal S''}(f), \zeta_{\mathcal S''}(w))\mid(f, w)\in \mu'',\  f \in F', w\in W'\cup {W''_{\beta}}\}.$$ 
    That is, $\zeta_{\mathcal S''}(\mu'')$ outputs a matching in $I'$ that contains all pairs of $\mu''$ defined over agents in $W'$, plus pairs $(f,w_j)$ with $f \in F'$, $w_j \in W$, and $w_j'' \in \mu''(f)$ for some copy $w_j'' \in {W''_\beta}$ of $w_j$ (recall that $W' \cap {W''_\beta}=\emptyset$). Finally, define $\xi_{\mathcal S''} = \xi_{\mathcal S'}\circ \zeta_{\mathcal S''}: {\mathcal S}'' \rightarrow \mathcal S$.
\hfill $\diamond$ \end{definition}

Note that it is not immediate that $\xi_{\mathcal S''}$ is well-defined, since for this to happen, the image of $\zeta_{\mathcal S''}({\cal S''})$ must be contained in $\mathcal S'$. We show that this is the case in the next lemma, whose proof is deferred to Section~\ref{sec:helpers}.

\begin{lemma}[{Stable matchings after the join constraint augmentation project to stable matchings}]\label{lem:join_constr_aug_projects}
    Let $\mu''\in \mathcal S''$. Then $\zeta_{\mathcal S''}(\mu'')$ is a stable matching in $\mathcal S'$. Hence, $\xi_{\mathcal S''}$ is well-defined.
\end{lemma}

We can show that $\xi_{\cal S''}''(\cal S'')$ restricts $\xi_{\cal S'}(\cal S')$ by removing all and only the stable matching of $I$ whose corresponding set of rotations does not satisfy the join constraint $(\alpha,\beta)$. See Section~\ref{sec:helpers} for the proof.

\begin{lemma}[Join constraint augmentation respects join constraint]\label{lem:join_constr_aug_respects_join_constr}
    Let $\mu\in \xi_{\mathcal S'}(\mathcal S')$. $\mu$ is in the image of $\xi_{\mathcal S''}({\cal S}'')$ if and only if it satisfies the join constraint $(\alpha, \beta)$. That is, if and only if $\psi_{{\cal S}}(\mu)$ satisfies $(\alpha,\beta)$.
\end{lemma}

We can also show that join constraint augmentations preserve the extendability of the instance. See Section~\ref{sec:helpers} for the proof.

\begin{lemma}[Join constraint augmentation creates extendable instance]\label{lem:join_constr_aug_preserves_partial_rep}
    The instance $I''$ after the join constraint augmentation is an extendable instance, with order-embedding function $\xi_{\mathcal S''} = \xi_{\mathcal S'}\circ \zeta_{\mathcal S''}$.
\end{lemma}

When $I'$ is standard we can bound the time needed to compute the output of a join constraint augmentation, and show that $I''$ is standard as well.

\begin{lemma}[Runtime of join constraint augmentation]\label{lem:jca_runtime}
Let $I'=(F', W', (\mathcal C'_f)_{f\in F'}, (\mathcal C'_w)_{w\in W'})$ be standard.
Then, $I''$ is standard and takes $O(|F|(|W|+|\alpha|)) = O(|\Pi({\cal S})|^4 + |\Pi({\cal S})|^2 |\alpha|)$ time to compute, where $|\alpha|$ is the number of arguments in $\alpha$. Also, each of the choice functions evaluates in polynomial time.
\end{lemma}
\begin{proof}
     By Lemma~\ref{lem:join_constr_aug_preserves_partial_rep}, agents from $I''$ have consistent and substitutable choice functions. Thus, it is enough to show that we add $O(|F|(|W|+|\alpha|))$ new agents and there is an algorithm that can be encoded in time $O(|F|(|W|+|\alpha|))$, which for each agent $a$ and set $T$, outputs $\mathcal C_a(T)$ in polynomial time.
     
     Following the procedure in Definition~\ref{def:join_constraint_augmentation}, Steps \ref{def:join_constraint_augmentation:aux} and \ref{def:join_constraint_augmentation:reg_workers} take $O(|W|)$ time, and add $O(|W|)$ new agents to the instance. Step \ref{def:join_constraint_augmentation:add_workers} takes $O(|F||\alpha|)$ time to encode the choice function for the new auxiliary worker $w_0$. Similarly, it takes $O(|W|)$ time to encode the choice function for the new auxiliary firm $f_0$ in Step \ref{def:join_constraint_augmentation:aux_firm_choice}. For Step \ref{def:join_constraint_augmentation:reg_worker_choice}, the algorithm computes choice functions for $O(|W|)$ new regular workers, each of which takes $O(|F|)$ time to implement as a preference list, for a total of $O(|F||W|)$ time. Step \ref{def:join_constraint_augmentation:reg_firm_choice} computes regular choice functions for $O(|F|)$ regular firms by updating the sets $A'_f$, each of which takes $O(|W|)$ time, for a total of $O(|F||W|)$. In total, this gives $O(|F|(|W|+|\alpha|))$ time. We can also verify from their definitions that each of the choice functions evaluates in $O(|F''|+|W''|+|\alpha|)$ time, which is polynomial in $|F|+|W|+|\alpha|$.
\end{proof}

We can now show that join constraints can be implemented iteratively on extendable instances.

\begin{lemma}[Join constraints can be implemented on extendable instances]\label{lem:jca_implemented}
Let $I'$ be a standard extendable instance with stable matching lattice $(\mathcal S',\succeq)$ and order embedding $\xi_{\mathcal S'}: \mathcal S'\to \mathcal S$. Let $(\alpha,\beta)$ be a join constraint. There exists an extendable instance $I''$ with order-embedding function $\xi_{\mathcal S''}: \mathcal S''\to \mathcal S$ such that $\mu\in \xi_{\mathcal S''}(\mathcal S'')$ if and only if $\mu\in \xi_{\mathcal S'}(\mathcal S')$ and $\mu$ satisfies $(\alpha, \beta)$.  Furthermore, $I''$ is standard and can be computed in $O(|F|(|W|+|\alpha|))$ time and has $O(|W|)$ more agents than $I'$.
\end{lemma}

\begin{proof}
    By Lemma~\ref{lem:join_constr_aug_preserves_partial_rep}, $I''$ is an extendable instance with order-embedding function $\xi_{\mathcal S''}$. By Lemma~\ref{lem:join_constr_aug_respects_join_constr}, the image $\xi_{\mathcal S''}(\mathcal S'')$ contains exactly the elements of $\xi_{\mathcal S'}(\mathcal S')$ which satisfy $(\alpha, \beta)$. By Lemma~\ref{lem:jca_runtime},  $I''$ is standard and can be computed in $O(|F|(|W|+|\alpha|))$ time. The bound on the number of additional agents follows from Definition~\ref{def:join_constraint_augmentation}.
\end{proof}

\paragraph{Conclusion of the proof.} We can now prove Theorem~\ref{thm:omega_extension}.

\begin{proof}[Proof of Theorem~\ref{thm:omega_extension}.]
    Let $\emptyset = \Omega_0\subset\Omega_1\subset\cdots\subset \Omega_k = \Omega$ be a sequence of subsets of $\Omega$, so that $|\Omega_{i}\setminus \Omega_{i-1}|
= 1$. Let $(\alpha_i, \beta_i)\in \Omega_i\setminus \Omega_{i-1}$ for each $i \in [k]$.
    By Lemma~\ref{lem:jca_implemented}, we can iteratively apply the join constraints $(\alpha_i, \beta_i)$ to $I$, so that after each iteration, the instance $I^i$ is a standard $\Omega_i$-extension of $I$.
    The final instance $I^* = I^k$ has $O(|\Omega||W|)$ additional agents and can be computed in $O(|\Omega||F|(|W|+|(\alpha,\beta)^*|))$ time.
    
    Finally, since $\psi_{\mathcal S}$ and $\xi_{\mathcal S^*}$ are both order-embeddings, $\psi_{\mathcal S^*} = \psi_{\mathcal S}\circ \xi_{\mathcal S^*}$ is an order-embedding.
    By Lemma~\ref{lem:join_constr_aug_order_embedding} (applied iteratively), the maximal and minimal elements of $\mathcal S$ are in the image of $\xi_{\mathcal S^*}$. Thus, their images under $\psi_{\mathcal S}$ are in the image of the composition $\psi_{\mathcal S^*} = \psi_{\mathcal S}\circ \xi_{\mathcal S^*}$.
    We conclude that $\psi_{\mathcal S^*}$ is a partial representation function for $(\mathcal S^*, \succeq)$.
\end{proof}

\subsubsection{Deferred proofs}\label{sec:helpers}

In this section, we present the proofs of Lemmas~\ref{lem:join_constr_aug_projects}, \ref{lem:join_constr_aug_respects_join_constr}, and \ref{lem:join_constr_aug_preserves_partial_rep}. First, we introduce helper results and features of GI instances. It is well known that the following holds.

\begin{proposition}[Opposition of interests~\citep{blair1988lattice}]
    Let $(F, W, (\mathcal C_f)_{f\in F}, (\mathcal C_w)_{w\in W})$ be a matching market instance such that all choice functions satisfy path-independence. Let $(\mathcal S, \succeq)$ be the stable matching lattice, and let $\mu, \mu'\in \mathcal S$. Then $\mu\succeq \mu'$ if and only if $\mathcal C_w(\mu(w)\cup\mu'(w)) = \mu'(w)$ for all $w\in W$. 
\end{proposition}

Hence, the previous lemma implies that, for stable matchings $\mu$ and $\mu'$, all workers prefer $\mu'$ if and only if all firms prefer $\mu$.

A first important feature of GI instances is that a firm is a worker's favorite partner if and only if the worker is the firm's least favorite partner, and vice versa.

\begin{proposition}[\citet{gusfield1989stable}]\label{prop:GI_opposite_prefs}
    Let $a,b$ be agents from the opposite sides of the market of $I_{(X,\succeq)}$. Then, agent $a$ is the first entry of $P_b$ (i.e., agent $b$'s most preferred partner) if and only if agent $b$ is the last entry of $P_a$ (i.e., agent $a$'s least preferred acceptable partner).
\end{proposition}

Furthermore, every stable matching in a GI instance is a perfect matching.

\begin{proposition}[\citet{gusfield1989stable}]\label{prop:GI_exact_matching}
    For every agent $a\in F\cup W$ and every stable matching $\mu\in \mathcal S$, $|\mu(a)| = 1$.
\end{proposition}

From now on, coherently with Section~\ref{sec:proof:thm:omega_extension}, we let $I'$ be an extendable instance of $I$, and $I''$ be the instance after applying the join constraint augmentation for $(\alpha,\beta)$ to $I'$. 
The next lemma follows immediately by Definition~\ref{def:extendable_instance} and Definition~\ref{def:join_constraint_augmentation}. 

\begin{lemma}\label{lem:obs:basic-properties-augmentation} 
    We have:
    \begin{enumerate}
    \item $W''=W' \uplus\{w_0\} \uplus {W''_\beta}$. That is, one auxiliary worker, and additional copies $W''_\beta$ are added to $W'$.
    \item $F''=F' \uplus \{f_0\}$. That is, only one firm is added to $F'$, and this is auxiliary.
    \item For each $f \in F$ and $T\subseteq W''$, $C''_f(T) \cap copy''(w)\neq \emptyset$ for at most one $w \in W$.
    \end{enumerate}
\end{lemma}

We next show that the preference lists after the join constraint augmentation satisfy substitutability and consistency. %The proof is deferred to Appendix~\ref{sec:app:helpers}.

\begin{lemma}[Join constraint augmentation satisfies substitutability and consistency]\label{lem:join_constr_aug_subs_cons}
    All agents in $I''$ have substitutable and consistent choice functions. 
\end{lemma}

\begin{proof}
    Because of Lemma~\ref{lem:obs:basic-properties-augmentation}, we consider the following cases. 
    
    \noindent    \textbf{Case 1}: Let $f\in F''\setminus F''_{aux}$. Thus,  $f \in F$. We know by definition that $f$ has a $(copy''(w_1),\dots, copy''(w_{k}),A_f'')$-regular choice function. 
    
    \noindent \emph{Substitutability}.  Let $w \in T \subseteq W''$ with $w\in \mathcal C''_{f}(T)$. Suppose first $w$ is a regular worker. Then $w\in copy''(w_i)$ for some $w_i\in W$, and $w_i$ is the most preferred worker in $P_f$ with copies in $T$. So, for any $T'\subseteq T$, $w_i$ is still the most preferred worker in $P_f$ with copies in $T'\cup\{w\}$. It follows that $w\in \mathcal C''_f(T'\cup\{w\})$. 
    
    Suppose now $w\in W''_{aux}$ is an auxiliary worker. Then one of the following happens: either there is a pair $(w_\ell, w)\in A''_f$, and the most preferred worker $w_i$ in $P_f$ with copies in $T$ satisfies $w_i\preceq_f w_\ell$, or $T$ is only composed of auxiliary workers and non-acceptable partners for $f$. In either case, for any $T'\subseteq T$, we also know that either the most preferred worker $w_{\hat \imath}$ in $P_f$ with copies in $T'\cup\{w\}$ satisfies $w_{\hat \imath}\preceq_f w_\ell$, or $T'\cup \{w\}$ is only made of auxiliary and non-acceptable partners for $f$. In both cases, $w\in \mathcal C''_f(T'\cup\{w\})$. We conclude that $\mathcal C''_f$ is substitutable. 

    \noindent \emph{Consistency.} Consider now sets $T,T'$ such that $\mathcal C_f''(T)\subseteq T'\subseteq T$.
    
    a) Suppose first $T$ contains the copy of at least one worker in $P_f$, and let $w_i$ be the most preferred worker $w_i$ in $P_f$ with copies in $T$.
    Let $w\in \mathcal C_f''(T)$. 
    Assume first that $w$ is a regular worker. Then, by definition, $w\in copy''(w_i)$.
    $w_i$ must thus also be the most preferred worker $w_i$ in $P_f$ with copies in $T'$, and so $w\in \mathcal C_f''(T')$.
    Assume instead $w$ is an auxiliary worker. By definition of the regular choice function, there must exist a pair $(w_\ell, w)\in A''_f$ where $w_i\preceq_f w_\ell$.
    Because $\mathcal C''_f(T)\subseteq T'$, the copies of $w_i$ selected from $T$ are also present in $T'$.
    Furthermore, because $T'\subseteq T$, no worker strictly preferred to $w_i$ in $P_f$ exists in $T'$.
    Thus, $w_i$ remains the most preferred regular worker with copies in $T'$.
    Since $w\in \mathcal C''_f(T)\subseteq T'$, the auxiliary worker $w$ is also present in $T'$.
    Because $w_i$ remains the top regular worker in $T'$ and the condition $w_i\preceq_f w_\ell$ still holds, the choice function will successfully select $w$ from $T'$ as well.
    Thus, in both cases, $w\in C''_f(T')$, establishing that $\mathcal C''_f(T)\subseteq \mathcal C''_f(T')$.
    Since $T'\subseteq T$ ensures that no additional preferred regular workers or triggered auxiliary workers can appear, we conclude that $\mathcal C''_f(T') = \mathcal C''_f(T)$.
    
    b) Now suppose $T$ contains only auxiliary workers and non-acceptable partners.
    Then $C''_f(T)$ contains all and only the auxiliary workers $w_{aux}$ in $T$ such that $(w_\ell,w_{aux})\in A''_f$ for some $w_\ell$, which by definition are all and only the auxiliary workers in $T'$ such that $(w_\ell,w_{aux})\in A''_f$ for some $w_\ell$.
    Thus, ${\mathcal C}''_f(T')=\mathcal C''_f(T)$. We conclude that $\mathcal C''_f$ is consistent.

    \noindent {\bf Case 2:} Let $f=f_0$ (i.e., $f$ is the newly added auxiliary firm) and fix sets of workers $T', T$ with $T'\subseteq T$. \\
    \emph{Substitutability.} Let $w\in \mathcal C''_{f}(T)$. If $w = w_0$, then by definition $w\in \mathcal C''_{f}(T'\cup\{w\})$ as well. If $w\ne w_0$, then $w_0 \notin T$. Hence $w_0\not\in T'$, and $w\in {W''}_{\beta}$. Thus $w\in \mathcal C''_{f}(T'\cup\{w\})$ as well. We conclude that $\mathcal C''_{f}$ is substitutable. \\
    \emph{Consistency.} Suppose $\mathcal C''_{f}(T)\subseteq T'$. If $w_0\in T$, then $\mathcal C''_{f}(T) = \mathcal C''_{f}(T') = \{w_0\}$. If $w_0\not\in T$, then $T'\supseteq \mathcal C''_{f}(T)=T\cap {W''}_{\beta}$ implies $\mathcal C''_{f}(T) = T' \cap {W''}_{\beta} = C''_{f}(T')$. Thus, $\mathcal C''_{f}$ is consistent.
    
    \noindent {\bf Case 3:} Let $f \in F''_{aux} \setminus \{f_0\}$. That is, $f$ is an auxiliary firm of $I'$. In this case, ${\mathcal C}_f'' = {\mathcal C}_f'$, and we know by hypothesis that the latter is substitutable and consistent.
    
    \noindent {\bf Case 4:} Let $w \in {W''_{\beta}}$. Then $w\in copy''(w_j)$ for some $w_j \in W$. The preference list $P''_{w}$ is a truncated version of the original preference list $P_{w_j}$, which is a preference list of agents, plus the singleton set $\{f_0\}$ in first position. Thus, the corresponding choice function $ P''_{w}$ is substitutable and consistent.
    
    \noindent {\bf Case 5:} Let $w = w_0$ (i.e., $w$ is the newly added auxiliary worker). Fix sets of firms $T',T$ with $T'\subseteq T$.  \\ 
    \emph{Substitutability.} Let $f\in \mathcal C''_{w}(T)$. If $f = f_0$, we know that $f_0\in \mathcal C''_{w}(T)$ implies $\alpha(\pi(F_\alpha \setminus T)) = 1$. Since $\alpha(\pi(F_\alpha \setminus T))$ is weakly decreasing in $T$ and $f_0 \notin F_\alpha$ since $F_\alpha \subseteq F$, we also know that $\alpha(\pi(F_\alpha \setminus (T'\cup\{f_0\}))) = 1$. Thus, $f_0\in \mathcal C''_{w}(T'\cup\{f_0\})$ as well. If $f\ne f_0$, then $f\in F_{\alpha}$, and so  $f\in \mathcal C''_{w}(T'\cup\{f\})$. Hence, $\mathcal C''_{w_0}$ is substitutable.
    
    \emph{Consistency.} Suppose $\mathcal C''_{w}(T)\subseteq T'$. By construction, we have $\mathcal C''_w(T)\cap F_\alpha = T\cap F_\alpha = T'\cap F_\alpha$. It follows that $\alpha(\pi(F_\alpha \setminus T)) = \alpha(\pi(F_\alpha \setminus T'))$. Thus, let $f \in {\mathcal C}''_{w}(T)\subseteq T'$. If $f \in F_\alpha$, then $f \in {\mathcal C}''_w(T')$ by construction. Else, $f = f_0$, and again $f \in {\mathcal C}''_w(T')$ because of construction and of $\alpha(\pi(F_\alpha \setminus T)) = \alpha(\pi(F_\alpha \setminus T'))$. To establish full equality, we must also confirm that $f_0$ is not newly selected from $T'$. If $f_0\not\in \mathcal C''_w(T)$, then either $f_0\not\in T$ (implying $f_0\not\in T'$) or $\alpha(\pi(F_\alpha \setminus T')) = \alpha(\pi(F_\alpha \setminus T)) = 0$. In either scenario, $f_0\not\in \mathcal C''_w(T')$. We conclude that ${\mathcal C}''_{w}(T') = {\mathcal C}''_{w}(T)$, and so $\mathcal C''_{w}$ is consistent.
    
    \noindent {\bf Case 6:} Let $w \in W''\setminus (\{w_0\} \cup {W''_{\beta}})$. In this case, ${\mathcal C}_w'' = {\mathcal C}_w'$, and we know by hypothesis that the latter is substitutable and consistent. 
\end{proof}

We can now prove Lemma~\ref{lem:join_constr_aug_projects}.

\begin{proof}[Proof of Lemma~\ref{lem:join_constr_aug_projects}.]
    Clearly, $\zeta_{\mathcal S''}({\mu''})$ is a matching in $I'$. 
    
    \noindent\emph{Individual rationality of $\zeta_{\mathcal S''}({\mu''})$.} 
    
    \noindent \textbf{Case 1}: Let $ w \in W'$. For any $w\in W'\setminus W$, ${\mu''}(w) = \zeta_{\mathcal S''}(\mu'')(w)$, and $\mathcal C'_w = \mathcal C''_w$. We claim the same is true for $w_j\in W$. By definition, ${\cal C}'_{w_j}={\cal C}''_{w_j}$. To show ${\mu''}(w_j) = \zeta_{\mathcal S''}(\mu'')(w_j)$, pick an arbitrary $w''_j \in copy''(w_j)$. We need to show that $\mu''(w_j'')\cap F' \subseteq \mu''(w_j)\cap F'$. By construction, only agents from $F$ (resp., $F \cup \{f_0\}$) are acceptable for $w_j$ (resp., $w''_j$). Thus, $\mu''(w_j'')\cap F'=\mu''(w_j'') \cap F $ and $\mu''(w_j)\cap F'=\mu''(w_j) \cap F $. Assume by contradiction that there exists $f \in \mu''(w''_j)\cap F \setminus \mu''(w_j) \cap F$. By definition of $\mathcal{C}''_f$, $f$ demands $w_j$ at $\mu''$. Thus, there exists $f' \in F''\setminus \{f\}$ such that $f' \in \mu''(w_j)$. Since only agents from $F$ are acceptable for $w_j$, we deduce that $f' \in F\setminus\{f\}$.

    We claim that there exists some $w_k \in W$ such that no agent from $copy''(w_k)$ is matched to any firm from $F$ in $\mu''$. 
    Indeed, by Proposition~\ref{prop:GI_exact_matching}, stable matchings in GI instances are perfect matchings, so $|F| = |W|$. Furthermore, Lemma~\ref{lem:obs:basic-properties-augmentation} and individual rationality of $\mu''$ imply that for each regular firm $f\in F$, $\mu''(f)\cap copy''(w) = \mathcal C''_f(\mu''(f))\cap copy''(w) \ne \emptyset$ for at most one $w\in W$. As $w_j$ and its copy $w_j''$ are matched to $f'\neq f \in F$, respectively, the claim follows. 
 
    As ${P}''_{w_k}={P}_{w_k}$ by construction, $w_k$ demands at $\mu''$ the last firm $f_k$ $P_{w_k}$. By Proposition~\ref{prop:GI_opposite_prefs}, $w_k$ is $f_k$'s most preferred partner in ${P}_{f_k}$. Thus, ${\cal C}''_{f_k}$ is a $(copy(w_k),\dots,A''_{f_k})$-regular choice function. Consequently, $(w_k, f_k)$ is a blocking pair for $\mu''$, a contradiction. We have thus shown ${\mu''}(w_j) = \zeta_{\mathcal S''}(\mu'')(w_j)$.
    Therefore we know that for all $w\in W'$, $\mathcal C'_w(\zeta_{\mathcal S''}({\mu''})(w)) = \mathcal C''_w({\mu''}(w)) = {\mu''}(w) = \zeta_{\mathcal S''}({\mu''})(w)$, where the second equality follows by individual rationality of $\mu''$. 
    
   \noindent \textbf{Case 2}: Let $f\in F'\setminus F$. By construction, $f$ is not acceptable for any worker from ${W''_\beta}$. Then ${\mu''}(f) = \zeta_{\mathcal S''}({\mu''})(f)$. Moreover, $\mathcal C'_f = \mathcal C''_f$ by construction. Thus we have again that $\mathcal C'_f(\zeta_{\mathcal S''}({\mu''})(f)) = \mathcal C''_f({\mu''}(f)) = {\mu''}(f) = \zeta_{\mathcal S''}({\mu''})(f)$. 
   
    \noindent \textbf{Case 3}: Let $f\in F$. Assume first ${\mu''}(f)\cap W''_{reg}\neq \emptyset$. By definition of $\mathcal C''_f$, and by individual rationality of ${\mu''}$, we know that ${\mu''}(f)\cap W''_{reg}\subseteq copy''(w_j)$ for some $w_j\in W$. Then, $\emptyset \neq \zeta_{\mathcal S''}({{\mu''}})(f)\cap W'_{reg}\subseteq copy'(w_j)$. Hence, by definition of $\mathcal C'_f$, $f$ demands $\zeta_{\mathcal S''}({{\mu''}})(f)\cap W'_{reg}$ at $\zeta_{\mathcal S''}({{\mu''}})(f)$. Now let $w_{aux}\in \zeta_{\mathcal S''}({{\mu''}})(f)\cap W'_{aux}$. By individual rationality of ${{\mu''}}$, we know that $(w_\ell, w_{aux})\in A''_f$ for some $w_\ell\succeq_f w_j$. Since $w_0\not\in W'$, we have $w_{aux}\ne w_0$, and so $(w_\ell, w_{aux})\in A'_f$ as well. By definition of $\mathcal C'_f$, firm $f$ demands $w_{aux}$ at $\zeta_{\mathcal S''}({{\mu''}})(f)$ as well. We conclude that $\mathcal C'_f(\zeta_{\mathcal S''}({{\mu''}})(f))= \zeta_{\mathcal S''}({{\mu''}})(f)$, and so $\zeta_{\mathcal S''}({{\mu''}})$ is individually rational. If ${{\mu''}}(f)\cap W_{reg}=\emptyset$, then $\zeta_{\cal S''}({{\mu''}})(f)\subseteq W'_{aux}$. Similarly to the argument above, we deduce that, for each $w'_{aux} \in \zeta_{\cal S''}({{\mu''}})(f)$, there exists $w \in W$ such that $(w,w'_{aux}) \in A'_f$. Thus, ${\cal C}'_f(\zeta_{\cal S''}({{\mu''}})(f))=\zeta_{\cal S''}({{\mu''}})(f)$. We conclude again that $\zeta_{\mathcal S''}({{\mu''}})$ is individually rational.
    
    \noindent\emph{Absence of blocking pairs in $\zeta_{\mathcal S''}({{\mu''}})$.} Assume for contradiction that $(f, w)$ is a blocking pair for $\zeta_{\mathcal S''}({{\mu''}})$. $w$ does not change their choice function when going from $I'$ to $I''$, that is, $\mathcal C_w' = \mathcal C_w''$. We consider two cases.

    \noindent \textbf{Case 1}: $f\in F'\setminus F$. That is, $f$ is an auxiliary firm in $I'$. $f$ does not change their preference list when going from $I'$ to $I''$: $\mathcal C_f' = \mathcal C_f''$. Furthermore, every acceptable partner for $f$ (resp., $w$) under $\mathcal C_f''$ (resp., $\mathcal C_w''$) is in $W'$ (resp., $F'$). So, ${{\mu''}}(f) = \zeta_{\mathcal S''}({{\mu''}})(f)$ and ${{\mu''}}(w) = \zeta_{\mathcal S''}({{\mu''}})(w)$. We see that if $f$ and $w$ demand each other in $\zeta_{\mathcal S''}({{\mu''}})$, they must also demand each other in ${{\mu''}}$, and thus form a blocking pair in ${{\mu''}}$, a contradiction.

    \noindent \textbf{Case 2}: $f\in F$. As before, ${{\mu''}}(w) = \zeta_{\mathcal S''}({{\mu''}})(w)$.
    We also see that $\zeta_{\mathcal S''}({{\mu''}})(f)$ is equal to ${{\mu''}}(f)$, possibly minus some workers $w''\in {W''_{\beta}}$ or $w_0$, plus at most the unique node $\hat{w}\in W$ such that $W''_\beta \cap copy''(\hat{w})\cap \mu''(f)\neq \emptyset$. We know by definition of $\mathcal C'_f$ and $\mathcal C''_f$, though, that any $\tilde{w}$ demanded by $f$ at $\zeta_{\mathcal S''}({{\mu''}})(f)$ under $\mathcal C'_f$ is also demanded by $f$ at ${{\mu''}}(f)$ under $\mathcal C''_f$. Thus, $(f, w)$ also blocks ${{\mu''}}$, a contradiction.

    So, a blocking pair $(f, w)$ for $\zeta_{\cal S''}({{\mu''}})$ does not exist. Thus, $\zeta_{\mathcal S''}({{\mu''}})$ is a stable matching in $\mathcal S'$.
\end{proof}

We next have a lemma relating  matchings from ${\cal S}''$ with their (stable) images under $\zeta_{\mathcal S''}$ and $\xi_{\mathcal S'} \circ \zeta_{\mathcal S''}$. 

\begin{lemma}[Augmented matchings contain original matchings]\label{lem:from-mu-to-muprime-to-muprimeprime}
    Let $\mu''\in \mathcal S''$ with $\mu' = \zeta_{\mathcal S''}(\mu'')$ and $\mu = \xi_{\mathcal S''}(\mu'') = \xi_{\mathcal S'}(\mu')$. Then, $\mu\subseteq \mu'\subseteq \mu''$. Furthermore, for $w\in W$ and $w'\in W'$, we have $\mu(w) = \mu'(w)=\mu''(w)$ and $\mu'(w') = \mu''(w')$.
\end{lemma}
\begin{proof}
    Recall that $\mu' \in {\cal S}'$ by Lemma~\ref{lem:join_constr_aug_projects} and $\mu \in {\cal S}$ by definition of extendable instance. We first show $\mu \subseteq \mu'$. Consider any $(f, w)\in \mu$. By Proposition~\ref{prop:GI_exact_matching},  $\mu(w) = \{f\}$, $\mu(f)=\{w\}$. Thus, $\mu'(w)\subseteq\{f\}$. By definition of the extendable instance $I'$, there exists a pair $(f, w')\in \mu'$ such that $w'\in copy'(w)$. By definition of the choice function $\mathcal C'_f$, we know that if $f$ demands $w'$ at a matching, it must demand all agents in $copy'(w)$, including $w$ itself. By definition of ${\cal C}'_w$ and since $\mu'(w)\subseteq \{f\}$, if $w'$ demands $f$ at $\mu'$, so does $w$. Since $\mu'$ is stable in $I'$, hence individually rational, $f,w'$ demand each other at $\mu'$. Thus, $f,w$ demand each other at $\mu'$ and by stability of $\mu'$ we know that $(f,w)\in\mu'$. We deduce that $\mu\subseteq\mu'$. 
    
    We next show $\mu' \subseteq \mu''$. Let $(f, w)\in \mu'$. We consider two cases.

    \textbf{Case 1}: $w\not\in W$. Then by definition of $\zeta_{\mathcal S''}$, we must have $(f, w)\in \mu''$ as well.

    \textbf{Case 2}: $w\in W$. Then there exists $(f, w')\in \mu''$ such that $w'\in copy''(w)$. 
    Since $\mu' \in {\cal S}'$ and the only acceptable sets for $w$ under the choice function $\mathcal C'_w$ are singletons in $F$, we know that $\mu'(w) = \{f\}$. It follows that $\mu''(w)\subseteq \{f\}$ as well, or $\mu'$ would also contain $(f', w)$ for some  $f'\in F\setminus \{f\}$, contradicting the individual rationality of $\mu'$. Thus, $w$ demands $f$ at $\mu''$. By definition of $\mathcal C''_f$, we know that $f$ demands all agents in $copy''(w)$ as well, including $w$. Thus, by stability of $\mu''$, we must have $(f,w)\in \mu''$.
\end{proof}

Now, we can prove Lemma~\ref{lem:join_constr_aug_respects_join_constr}.

\begin{proof}[Proof of Lemma~\ref{lem:join_constr_aug_respects_join_constr}.]
    First, suppose $\mu\in \xi_{\mathcal S'}(\mathcal S')$ satisfies the join constraint. Let $\mu'\in \mathcal S'$ be such that $\xi_{\mathcal S'}(\mu') = \mu$. Recall that $\mu\in \mathcal S$ by definition of extendable instance. Thus, by Proposition~\ref{prop:GI_exact_matching}, we know that, for each $f \in F$, $\mu(f) = \{w\}$ for some $w\in W$. Consider the matching $\mu''$ in $I''$ given by the following procedure:
    \begin{enumerate}
        \item First, set $\mu''\gets \mu'$.\label{lem:join_constr_aug_respects_join_constr:step_mu'}
        \item For each $f\in F_{\alpha}$, let $w\in P_f$ be the unique worker such that $(f, w)\in \rho^-$ for some $\rho\in \Pi_\alpha$ (see Proposition~\ref{prop:GI_rotations_overlap}). If $\mu(f) \preceq_f w$, then add $(f, w_0)$ to $\mu''$. \label{lem:join_constr_aug_respects_join_constr:step_rho}
        \item If $\alpha(\pi(F_\alpha\setminus\mu''(w_0))) = 1$, then add to $\mu''$:\label{lem:join_constr_aug_respects_join_constr:step_if}
        \begin{enumerate} \item $(f_0,w_0)$;
        \item $(\mu(w), {w''})$ for each $w \in W_\beta$ and ${w''\in W''_{\beta}}\cap copy''(w)$;
        \end{enumerate}
        \item Else, add $(f_0, {w''})$ to $\mu''$, for each ${w''\in W''_{\beta}}$.

    \end{enumerate}
    Our goal is to show that $\mu''$ is stable in $I''$, and that $\zeta_{\mathcal S''}(\mu'') = \mu'$. 
    
    \noindent\emph{Individual rationality of $\mu''$.} First we show individual rationality of $\mu''$. It is easily verified by inspection that, for $a\in F''\cup W''\setminus {W''_\beta}$, $\mu''(a) = \mathcal C''_a(\mu''(a))$. Now consider ${w''\in W''_\beta}$, and let $w'' \in  copy''(w)$ for some $w \in W$. Note that either $\mu''({w''})=\mu(w)$ (if the ``if'' condition from Step~\ref{lem:join_constr_aug_respects_join_constr:step_if} above applies) or $\mu''({w''})=\{f_0\}$ (if it does not). Suppose the latter holds. Then, by Definition~\ref{def:join_constraint_augmentation}, $\mu''({w''})= \mathcal C''_{w''}(\mu''({w''}))$. 
    
    Now suppose $\mu''({w''}) = \mu(w)$. Since by definition the preference list $P''_{w''}$ is given by a list of firms, it suffices to show that $\mu''({w''})$ belongs to the list. From Step~\ref{lem:join_constr_aug_respects_join_constr:step_if} above, we must have $\alpha(\pi(F_{\alpha}\setminus\mu''(w_0))) = 1$. 
    
    Let $f \in F_\alpha$. Fix $\rho \in \Pi_\alpha$ so that $f \in F_\rho$, and let $\hat w$ be the unique worker such that $(f,\hat w)\in \rho^-$. By Step~\ref{lem:join_constr_aug_respects_join_constr:step_rho}, $f\in\mu''(w_0)$ if and only if $\mu(f)\preceq_f \hat w$, which happens if and only if $\rho\notin \psi_{\mathcal S}({\mu})$ (recall that $F_\rho \cap F_{\rho'}=\emptyset$ for $\rho\neq \rho' \in \Pi_\alpha$ by Lemma~\ref{obs:F-rho-do-not-intersect}).
    Thus, $\pi(F_\alpha\setminus \mu''(w_0)) = \{\rho \in \Pi_\alpha \mid F_\rho \subseteq F_\alpha\setminus \mu''(w_0)\} = \{\rho \in \Pi_\alpha \mid F_\rho \subseteq F_\alpha\setminus \cup_{\rho \in \Pi_\alpha : \rho \notin \psi_{\cal S} ({\mu})} F_\rho  \} = \{ \rho \in \Pi_\alpha \mid \rho \in \psi_{\cal S}({{\mu}})\} = \Pi_\alpha \cap \psi_{\cal S}({\mu})$.
    % $$\begin{array}{lll} \pi(F_\alpha\setminus \mu''(w_0)) & = & \{\rho \in \Pi_\alpha \mid F_\rho \subseteq F_\alpha\setminus \mu''(w_0)\} \\[1.2ex]
    % & = & \{\rho \in \Pi_\alpha \mid F_\rho \subseteq F_\alpha\setminus \cup_{\rho \in \Pi_\alpha : \rho \notin \psi_{\cal S} ({\mu})} F_\rho  \} \\[1.2ex] 
    % & = & \{ \rho \in \Pi_\alpha \mid \rho \in \psi_{\cal S}({{\mu}})\} \\[1.2ex] 
    % & = & \Pi_\alpha \cap \psi_{\cal S}({\mu}). \end{array}$$
    We deduce $\alpha(\psi_{\cal S}({{\mu}}))=\alpha(\Pi_\alpha \cap \psi_{\cal S}({{\mu}}))=\alpha(\pi(F_{\alpha}\setminus\mu''(w_0))) = 1.$ Since $\mu$ satisfies the join constraint $(\alpha, \beta)$ by hypothesis, we also have $\beta(\psi_{\mathcal S}({{\mu}})) = 1$. Thus, $\rho\in \psi_{\mathcal S}({{\mu}})$ for all $\rho\in \Pi_\beta$. 
    
    Recall that we have fixed ${w'' \in W''_\beta} \cap copy''(w)$. Now let $\hat f$ be the least preferred firm in $P_{w}$ such that $(\hat f, w)\in\rho^+$ for some $\rho\in \Pi_{\beta}$, as in Definition~\ref{def:join_constraint_augmentation} Step~\ref{def:join_constraint_augmentation:reg_worker_choice}. Then we know that $\mu(w)\preceq_{w}\hat f$, since the corresponding rotation $\rho\in \Pi_\beta$ belongs to $\psi_{\cal S}({{\mu}})$. 
    We deduce that $\mu''({w''})=\mu(w)\preceq_{w}\hat f$. As a result, $\mu''({w''})$ is an element of the preference list $P''_{{w''}}$, as required. %and so $\mu''({w''_j}) = \mathcal C''_{w''_j}(\mu''({w''_j}))$. 
    We conclude that $\mu''$ is individually rational.
    
    \noindent\emph{Absence of blocking pairs in $\mu''$.} Let $(f,w) \in (F'' \times W'')\setminus \mu''$. We consider several cases. In all of them, we show that either $f$ does not demand $w$ at $\mu''$, or $w$ does not demand $f$ at $\mu''$, or that $(f,w)$ blocks $\mu$ or $\mu'$ as well, a contradiction. Thus, $(f,w)$ does not block $\mu''$.

    \noindent    \textbf{Case 1}: $f = f_0$ and $w = w_0$. Then $\alpha(\pi(F_\alpha\setminus\mu''(w_0))) = 0$. By construction of $\mathcal C''_{w_0}$, $w_0$ does not demand $f_0$ at $\mu''$.

    \noindent    \textbf{Case 2}: $f = f_0$ and $w\ne w_0$. Then $w\in {W''_{\beta}}$, as other than $w_0$ only the workers in ${W_{\beta}''}$ are acceptable to $f_0$. By construction, either $\mu''(f_0)={W''}_\beta$, or $\mu''(f_0) = \{w_0\}$. By definition the latter must hold, but then by construction $f_0$ does not demand $w$ at $\mu''$. 

    \noindent      \textbf{Case 3}: $f\ne f_0$ and $w = w_0$. Then $f\in F_{\alpha}$, as other than $f_0$ only the firms in $F_{\alpha}$ are acceptable to $w_0$. Let $w_i\in P_f$ be the worker such that $(f, w_i)\in \rho^-$ for some $\rho\in \Pi_\alpha$, which is unique by Proposition~\ref{prop:GI_rotations_overlap}. Then $\mu(f) = w_j\succ_f w_i$, because if $\mu(f)\preceq_f w_i$ held, then $(f, w_0)$ would have been added to $\mu''$ in Step~\ref{lem:join_constr_aug_respects_join_constr:step_rho}. Thus, $w_j'\in \mu'(f)$ for some $w_j'\in copy'(w_j)$. Step~\ref{lem:join_constr_aug_respects_join_constr:step_mu'} in the construction of $\mu''$ implies that $w'_j \in \mu''(f)$. By the definition of the join constraint augmentation, $(w_i, w_0)\in A''_f$, and so under the choice function $\mathcal C''_f$, firm $f$ does not demand $w_0$ at $\mu''$, since it is matched to $w_j'$ and $w_j \succ_f w_i$.

    \noindent    \textbf{Case 4}: $f\in F'\setminus F$. Thus, in particular, $f \notin F_\alpha$ (since $F_\alpha \subseteq F$) and $w\in W'$, as no worker from {$W''_\beta$} is acceptable for $f$ under $\mathcal C''_f$. Hence, the only step in the construction of $\mu''$ that involves $f$ or $w$ is 1. Therefore, $\mu'(f) = \mu''(f)$ and $\mu'(w) = \mu''(w)$. Furthermore, $\mathcal C''_f = \mathcal C'_f$ and $\mathcal C''_w = \mathcal C'_w$. So, if $(f, w)$ blocks $\mu''$, it must also block $\mu'$.

    \noindent    \textbf{Case 5}: $f\in F$ and $w\in {W'}$. As in the previous case, we know that $\mu'(w) = \mu''(w)$ and $\mathcal C''_w = \mathcal C'_w$. Thus, if $w$ demands $f$ at $\mu''$, it also demands $f$ at $\mu'$. Also, $\mu''(f)$ coincides with $\mu'(f)$, plus possibly $w_0$ and new copies of $\hat w \in W \cap \mu'(f)$. Assume $f$ demands $w$ at $\mu''$. Then by construction $f$ demands $w$ at $\mu'$. Hence, if $(f,w)$ blocks $\mu''$, it also blocks $\mu'$.

    \noindent    \textbf{Case 6}: $f\in F$ and $w = {w_j''}\in {W''_{\beta}}$, where ${w_j''}\in copy''(w_j)$ for some $w_j \in W$. Then $f_0\notin\mu''(w)$ since, by construction, if $f_0 \in \mu''(w)$, $w$ does not demand any other firm. By Lemma~\ref{lem:from-mu-to-muprime-to-muprimeprime}, we know that $\mu'(w_j)=\mu(w_j)$. Therefore, by construction of $\mu''$ and the definition of $\xi_{\mathcal S'}:\mathcal S'\to \mathcal S$, we have $\mu''({w_j''}) = \mu(w_j)=\mu'(w_j)=\mu''(w_j)$. Thus $(f, w_j)\not\in \mu'$ as well. If $f$ demands ${w_j''}$ at $\mu''(f)$ under $\mathcal C''_f$, then $f$ also demands $w_j$ at $\mu'(f)$ under $\mathcal C'_f$. If ${w_j''}$ demands $f$ at $\mu''({w_j''})$ under $\mathcal C''_{{w_j''}}$, then by construction of $\mathcal C''_{{w_j''}}$, $w_j$ demands $f$ at $\mu'(w_j)$ under $\mathcal C'_{w_j}$ as well. Thus, $(f, w_j)$ blocks $\mu'$, a contradiction.

    We see that there can be no blocking pair, and so $\mu''$ is stable in $I''$. Furthermore, it is easy to verify that $\zeta_{\mathcal S''}(\mu'') = \mu'$, which then implies $\xi_{\mathcal S''}(\mu'') = \mu$. 

    \smallskip 

    For the other direction, let $\mu = \xi_{\mathcal S''}(\mu'')$ for some $\mu''\in \mathcal S''$ and suppose that $\alpha(\psi_{\mathcal S}({{\mu}})) = 1$. Our goal is to show that $\beta(\psi_{\mathcal S}({{\mu}})) =1$. Let $w \in W$,  $w'' \in W''_\beta \cap copy''(w)$. We first claim that $(f_0,w'')\notin \mu''$. 
    
    Consider any $\rho\in \Pi_\alpha\cap \psi_{\mathcal S}({{\mu}})$, and any $(f, \hat w)\in \rho^-$. Since $\rho$ has occurred in $\mu$, we know that $\mu(f)\succ_{f} \hat w$. We also know by Lemma~\ref{lem:from-mu-to-muprime-to-muprimeprime} that $\mu(f)\subseteq \mu''(f)$, with $\mu''(f)$ possibly also containing edges incident to $w_0$ and other elements of $copy''(\hat w)$. However, by construction of $\mathcal C''_f$, firm $f$ does not demand $w_0$ at $\mu(f)$. By substitutability, firm $f$ also does not demand $w_0$ at $\mu''(f)\supseteq \mu(f)$, and so $w_0\not\in \mu''(f)$ by individual rationality of $\mu''$. Since this is true for all $\rho\in \Pi_\alpha\cap \psi_{\mathcal S}({{\mu}})$ and all $(f, \hat w)\in \rho^-$, we see that, by monotonicity of $\alpha (\cdot )$, we have 
    $\alpha(\pi(F_\alpha\setminus\mu''(w_0))) \geq \alpha(\Pi_\alpha \cap \psi_{\cal S}({{\mu}})) = \alpha((\psi_{\cal S}({{\mu}}))=1.$
    Thus, $w_0$ demands $f_0$ at $\mu''(w_0)$. Since $f_0$ always demands $w_0$ and $\mu''$ is stable, we have $(f_0, w_0)\in \mu''$. By construction of $\mathcal C''_{f_0}$ and individual rationality of $\mu''$, we deduce $(f_0, {w''})\not\in \mu''$.
    
    We show next that $w''$ is matched in $\mu''$. Assume this is not the case. Then, ${w''}$ demands the last firm $f_k$ on their preference list $P''_{{w''}}$, which is also the last firm on $P_{w}$. By Proposition~\ref{prop:GI_opposite_prefs}, we know that $w$ is the first worker on the preference list $P_{f_k}$. By definition, $f_k$ also always demands all workers in $copy''(w)$ under the choice function $\mathcal C''_f$. Thus, $(f_k, {w''})$ would block $\mu''$, a contradiction. Thus, $w''$ is matched in $\mu''$.

    Since we argued above that $(f_0,w'') \notin \mu''$, it follows that $\mu''({w''}) = \{f_\ell\}$ for some firm $f_\ell\ne f_0$ in $P''_{{w''}}$. Let $f_j$ be the least preferred firm in $P_{w}$ such that $(f_j, w)\in\rho^+$ for some $\rho\in \Pi_\beta$. Then, by construction of $P''_{{w''}}$, the only acceptable firms for ${w''}$ other than $f_0$ are those less preferable or equal to $f_j$ under $P_{w}$. Thus, by stability of $\mu''$ and Lemma~\ref{lem:from-mu-to-muprime-to-muprimeprime}, we have $\mu''({w''}) = \mu''(w) = \mu(w) \preceq_{w} \{f_j\}$. 

    Since $w \in W$, $w'' \in W''_\beta \cap copy''(w)$ were selected arbitrarily, it follows that all rotations in $\Pi_\beta$ have occurred in $\mu$. Equivalently, $\Pi_\beta\subseteq\psi_{\mathcal S}({{\mu}})$, which implies $\beta(\psi_{\mathcal S}({{\mu}})) = 1$. 
\end{proof}

Now, we have a result showing that the new stable matching lattice admits a natural order-embedding to the previous stable matching lattice. 

\begin{lemma}[Join constraint augmentation admits order-embedding]\label{lem:join_constr_aug_order_embedding}
    $\zeta_{\mathcal S''}:\mathcal S''\to\mathcal S'$ is an order-embedding, under the standard partial orders over stable matchings. As a consequence, $\xi_{\mathcal S''}:\mathcal S''\to\mathcal S$ is also an order-embedding. Furthermore, the maximal and minimal elements of $\mathcal S$ are in the image of $\xi_{\mathcal S''}$.
\end{lemma}
\begin{proof}
    To distinguish between the partial orders in ${\cal S}'$ and ${\cal S}''$, we will use $\succeq'$ and $\succeq''$, respectively.
    Let $\mu''_1, \mu''_2\in \mathcal S''$. First, consider any firm $f$ in $F'\setminus F$.
    Then $\mathcal C''_f = \mathcal C'_f$,
    $\mu''_1(f) = \zeta_{\mathcal S''}(\mu''_1)(f)$, and $\mu''_2(f) = \zeta_{\mathcal S''}(\mu''_2)(f)$.
    Thus, $\mu''_1(f)\succeq''_f \mu''_2(f)$ if and only if $\zeta_{\mathcal S''}(\mu''_1)(f)\succeq_f'\zeta_{\mathcal S''}(\mu''_2)(f)$.

    Next, consider any firm $f\in F$.
    We know by Lemma~\ref{lem:from-mu-to-muprime-to-muprimeprime} and the definition of $\zeta_{\mathcal S''}$ that for $i=1,2$, $\zeta_{\mathcal S''}(\mu''_i)(f) = \mu''_i(f)\setminus (W''_{\beta}\cup \{w_0\})$.
    By construction of the standard choice function $\mathcal C''_f$ in Definition~\ref{def:join_constraint_augmentation}, $\mu''_1(f)\succeq''_f \mu''_2(f)$ if and only if $\mu''_1(f)\setminus (W''_{\beta}\cup \{w_0\})\succeq'_f \mu''_2(f)\setminus (W''_{\beta}\cup \{w_0\})$ as well, which is true if and only if $\zeta_{\mathcal S''}(\mu''_1)(f)\succeq'_f\zeta_{\mathcal S''}(\mu''_2)(f)$.
    Thus, $\mu''_1(f)\succeq''_f \mu''_2(f)$ if and only if $\zeta_{\mathcal S''}(\mu''_1)(f)\succeq_f'\zeta_{\mathcal S''}(\mu''_2)(f)$.

    Finally, consider any firm $f\in F''\setminus F'$, and suppose $\zeta_{\mathcal S''}(\mu''_1)\succeq'\zeta_{\mathcal S''}(\mu''_2)$.
    We show that $\mu''_1(f)\succeq''_f \mu''_2(f)$, which concludes the proof of the order embedding.
    We know $f = f_0$, and each of $\mu''_1(f_0), \mu''_2(f_0)$ is  either equal to $\{w_0\}$ or contained in $W''_\beta$.
    If $\mu''_1(f_0) = \{w_0\}$, then $\mu''_1(f_0)\succeq''_{f_0} \mu''_2(f_0)$, and we are done. Now suppose $\mu''_1(f_0)\ne\{w_0\}$.
    Then $f_0$ demands all agents in $W''_\beta$, who also by definition always demand $f_0$, and so $\mu''_1(f_0) = {W''_\beta}$ and $\alpha(\pi(F_\alpha\setminus\mu''_1(w_0))) = 0$ by stability and definition of ${\mathcal C}''_{w_0}$.
    Suppose for contradiction that $\mu''_1(f_0) \prec''_{f_0} \mu''_2(f_0)$. Then, since $\mu''_2$ is also stable, we have $\mu''_2(f_0) = \{w_0\}$.
    By individual rationality of $\mu''_2$ and the definition of ${\mathcal C}''_{w_0}$, we have $\alpha(\pi(F_\alpha\setminus\mu''_2(w_0))) = 1$.
    By monotonicity of $\alpha$, there exists some $f_j\in F_\alpha$ with $f_j\in \mu''_1(w_0)\setminus \mu''_2(w_0)$.
    By Proposition~\ref{prop:GI_exact_matching} there is a unique worker $w_1\in W$ such that $w_1\in \xi_{\mathcal S''}(\mu_1'')(f_j)\in\mathcal S$, and so by definition of $\mathcal \xi_{\mathcal S''}$, we also know that $w_1$ is the only element from $W$ such that $\mu''_1(f_j)\cap copy''(w_1) \ne\emptyset$.
    Similarly, there is a unique worker $w_2\in W$ such that $\mu''_2(f_j)\cap copy''(w_2) \ne\emptyset$.
    By definition of the standard choice function $C''_{f_j}$ and  because $f_j$ demands $w_0$ at $\mu''_1(f_j)$ by individual rationality of $\mu''$, we know that $w_1\preceq_{f_j} w_i$, where $(w_i, w_0)$ is the pair added to $A''_f$ in the join constraint augmentation.
    On the other hand, since $f_j$ does not demand $w_0$ at $\mu''_2(f_j)$ by our selection of $f_j$, we know that $w_2\succ_{f_j} w_i$, where $(w_i, w_0)$ is again the pair added to $A''_f$ in the join constraint augmentation.
    Thus, $\mu''_1(f_j)\prec''_{f_j}\mu''_2(f_j)$, and $\xi_{\mathcal S''}(\mu''_1)(f_j)\prec'_{f_j}\xi_{\mathcal S''}(\mu''_2)(f_j)$, contradicting $\zeta_{\mathcal S''}(\mu''_1)\succeq'\zeta_{\mathcal S''}(\mu''_2)$ and the fact that $\xi_{\mathcal S'}$ is an order embedding.
    We see that 
    $\mu''_1(f_0)\succeq''_{f_0} \mu''_2(f_0)$.

    Finally, consider the maximal element $\mu_F\in \mathcal S$.
    Then $\psi_{\mathcal S}(\mu_F) = \Pi$. So, for any $\rho\in \Pi$, we have $\mathbf{1}_\rho(\psi_{\mathcal S}(\mu_F)) = 1$.
    It follows that $\alpha(\psi_{\mathcal S}(\mu_F)) = \beta(\psi_{\mathcal S}(\mu_F))= 1$, which means $\mu_F$ satisfies $(\alpha, \beta)$, and thus $\mu_F\in\xi_{\mathcal S''}(\mathcal S'')$ by Lemma~\ref{lem:join_constr_aug_respects_join_constr}.
    Similarly, for the minimal element $\mu_W\in \mathcal S$, we have $\psi_{\mathcal S}(\mu_W) = \emptyset$.
    For any $\rho\in \Pi$, we have $\mathbf{1}_\rho(\psi_{\mathcal S}(\mu_W)) = 0$, and thus $\alpha(\psi_{\mathcal S}(\mu_W)) = \beta(\psi_{\mathcal S}(\mu_W)) = 0$.
    We conclude that $\mu_W$ satisfies $(\alpha,\beta)$ as well, and is in the image of $\xi_{\mathcal S''}$.
\end{proof}

Last, we prove Lemma~\ref{lem:join_constr_aug_preserves_partial_rep}.

\begin{proof}[Proof of Lemma~\ref{lem:join_constr_aug_preserves_partial_rep}.] We verify the properties from Definition~\ref{def:extendable_instance}. 

    \ref{def:extendable_instance:subs_cons}. We know by Lemma~\ref{lem:join_constr_aug_subs_cons} that all agents in $I''$ have substitutable and consistent choice functions.

    \ref{def:extendable_instance:partition}. We know by definition of the join constraint augmentation that the agents in $F''\cup W''$ can be partitioned into firms $F''_{aux}\uplus F$ and workers $W''_{aux}\uplus W''_{reg}$, and the set of regular workers $W''_{reg}$ can be partitioned into sets $\biguplus_{w\in W}copy''(w)$.

    \ref{def:extendable_instance:reg_workers}. The preference list $P''_w$ for $w\in W'$ are unchanged from $P'_w$, which already satisfied \ref{def:extendable_instance:reg_workers}.~by assumption. For agents $w''_j\in W''_{\beta}$, by definition the preference list $P''_{w''_j}$ is also of the form $(f_0, f_\ell, f_{\ell+1}, \dots, f_k)$, as desired.

    \ref{def:extendable_instance:reg_firms}. By definition, for every $f\in F$ the choice function $\mathcal C''_f$ is a standard choice function.

    \ref{def:extendable_instance:projection}. We know by Lemma~\ref{lem:join_constr_aug_order_embedding} that $\xi_{\mathcal S''}$ is an order-embedding, and that the maximal and minimal elements of $\mathcal S$ are in the image of $\xi$.
\end{proof}

\section{Proof of Theorem~\ref{thm:main1}}\label{sec:proof:thm:main1}
 
Let $(X,\succeq)$ be a lattice. By Theorem~\ref{thm:gusfield_irving} applied to $(X_j,\succeq)$, in $O(|X_j|^2) = O(|X|^2)$ time we can construct a GI instance $I=I_{(X_j, \succeq)}$  whose rotation poset $(\Pi(\mathcal S), \succeq)$ is order-isomorphic to the poset of join-irreducible elements $(X_j, \succeq)$ via $\phi: X_j\to \Pi(\mathcal S)$. By Corollary~\ref{cor:GI_repr}, $\phi^{-1}\circ\psi_{\mathcal S}: (\mathcal S, \succeq) \rightarrow (\mathcal D(X_j), \supseteq)$ is an order-isomorphism, where $({\cal S},\succeq)$ is the lattice of stable matchings of $I$. By Lemma~\ref{lem:poly_join_constrs}, we can compute in $O(|X|^3)$ time a set $\Omega$ of $O(|X|^2)$ join constraints on $(X_j, \succeq)$, such that the canonical partial representation $\psi_X: X\to \mathcal D(X_j)$ is an order-isomorphism from $(X, \succeq)$ to 
$$(\{T\in \mathcal D(X_j)\mid T \text{ satisfies } (\alpha,\beta)\ \forall (\alpha, \beta)\in\Omega\}, \supseteq).$$ Setting $\Omega_{\mathcal S} = \{(\alpha\circ \phi^{-1},\beta\circ \phi^{-1}) \mid (\alpha, \beta)\in\Omega\}$ as a set of join constraints over $(\Pi(\mathcal S), \succeq)$, we deduce that $\psi_{\mathcal S}^{-1}\circ\phi\circ\psi_X$ is an order-isomorphism from $(X, \succeq)$ to 
\begin{equation}\label{eq:mu-and-omega}(\{\mu\in \mathcal S\mid \mu \text{ satisfies } \omega,\ \forall \omega\in\Omega_{\mathcal S}\}, \succeq).\end{equation} 
By Theorem~\ref{thm:omega_extension}, we can construct a standard $\Omega_S$-extension $I^*$ of $I$ such that~\eqref{eq:mu-and-omega} is order-isomorphic to $(\mathcal S^*, \succeq)$, where $(\mathcal S^*,\succeq)$ is the lattice of stable matchings of $I^*$. Thus, $(\mathcal S^*, \succeq)$ is order-isomorphic to $(X, \succeq)$, as desired.
    
For the runtime, first observe that, by Theorem~\ref{thm:gusfield_irving}, $I$ has $O(|X|^2)$ agents and can be constructed in time $O(|X|^2)$. By Lemma~\ref{lem:poly_join_constrs} we have $|\Omega| = O(|X|^2)$, and by Theorem~\ref{thm:omega_extension} and Lemma~\ref{lem:poly_join_constrs}, it takes $O(|\Omega||F|(|W|+|(\alpha,\beta)^*|)) = O(|X|^6)$ time to construct $I^*$, which then has $O(|\Omega||W|) = O(|X|^4)$ agents. \hfill $\square$

\section{Antimatroids and complements of join constraints}\label{sec:hardness}

We now introduce tools that allow us to show that optimizing a linear function over the set of stable matchings in a standard matching market is NP-hard via a reduction from the problem of finding a feasible set of minimum cost in an antimatroid.

\subsection{Basic definitions}

\begin{definition}[Antimatroid]
    Let $V$ be a finite ground set and let $\mathcal G$ be a nonempty family of subsets of $V$, called \emph{feasible}. $(V, \mathcal G)$ is an \emph{antimatroid} if the following conditions hold:
    \begin{enumerate}
        \item ({\bf Feasibility of the ground set}) $V\in \mathcal G$.
        \item ({\bf Closure under union}) $G_1, G_2\in \mathcal G\implies G_1\cup G_2\in \mathcal G$. 
        \item ({\bf Accessibility}) $G\in \mathcal G\setminus\{\emptyset\} \implies \exists g\in G : G\setminus\{g\}\in\mathcal G$. \hfill $\diamond$
    \end{enumerate}
 \end{definition}

It is well-known that the feasible sets of an antimatroid, when ordered by inclusion, form a lattice $(\mathcal G, \supseteq)$, whose join operation is the union (see, e.g.,~\citet{korte2012greedoids}). 

\begin{definition}[Paths and endpoints]\label{def:path_endpoint}
    Let $(V,{\cal G})$ be an antimatroid. An \emph{endpoint} of $G \in {\cal G}$ is an element $g\in G$ such that $G\setminus\{g\}\in \mathcal G$. A \emph{path} is a feasible set $G\in \mathcal G$ with exactly one endpoint, denoted by $e(G)$. We let $\mathcal Q$ be the collection of all paths of $(V,{\cal G})$. Given any $x\in V$, we define $\mathcal Q(x) = \{G\in \mathcal Q\mid e(G) = x\}$. The set of paths $\mathcal Q$ can be partially ordered to form the \emph{path poset} $(\mathcal Q, \supseteq)$.
\hfill $\diamond$ \end{definition}

It is showed in~\citet{korte2012greedoids} that the paths of an antimatroid are exactly the feasible sets which are not the union of other feasible sets. In other words, paths are the join-irreducible elements of the lattice $(\mathcal G, \supseteq)$.

\begin{proposition}[Paths are join-irreducible elements \citep{korte2012greedoids}]\label{prop:paths_jie}
    A feasible set is a path if and only if it is not the union of two other feasible sets.
\end{proposition}

It is also showed in \citet{korte2012greedoids} that the lower closed sets of paths contain all endpoints.

\begin{proposition}[Path subsets contain all endpoints \citep{korte2012greedoids}]\label{prop:paths_contain_all_endpoints}
    Let $G\in\mathcal G$ be a feasible set, and let $g\in G$. Then, there is a path $G'\subseteq G$ such that $g$ is the endpoint of $G'$.
\end{proposition}

Similar to Birkhoff's representation theorem for lattices, $(\mathcal Q, \supseteq)$ generates all feasible sets of an antimatroid. 

\begin{proposition}[Feasible sets as unions of paths \citep{korte2012greedoids}]\label{prop:feasible_sets_representation}
    Let $(V, \mathcal G)$ be an antimatroid with path poset $(\mathcal Q, \supseteq)$. Define $\psi_{\mathcal G}: \mathcal G\to 2^{\mathcal Q}$ by $\psi_{\mathcal G}(G) = \{G'\in {\mathcal Q}\mid G'\subseteq G\}$. Then, for any $G\in \mathcal G$, we have $G = \bigcup\psi_{\mathcal G}(G)$. Conversely, for any ${\mathcal Q}'\subseteq {\mathcal Q}$, we have $\cup {\mathcal Q}'\in {\cal G}$.
\end{proposition}

\begin{example}
    Consider the antimatroid and its corresponding path poset shown in Figure~\ref{fig:antimatroid}. The ground set is $V = \{a, b, c, d\}$. We can verify that each path is not the join (i.e. union) of any other feasible sets, and has a unique endpoint. For example, for the path $\{a, c, d\}$, we can check that $\{a, c, d\}\setminus\{d\}$ is feasible (and thus $d$ is an endpoint), but $\{a, c, d\}\setminus\{a\}$ and $\{a, c, d\}\setminus\{c\}$ are not. We can also see that each element in $\{a, c, d\}$ is the endpoint of a path contained in $\{a, c, d\}$: $a$ is the endpoint for the path $\{a\}$, $c$ is the endpoint for $\{a, c\}$, and $d$ is the endpoint for $\{a, c, d\}$. \hfill $\triangle$

    \begin{figure}
        \centering
        \begin{subfigure}[t]{.4\linewidth}
        \centering\includegraphics[scale=0.5]{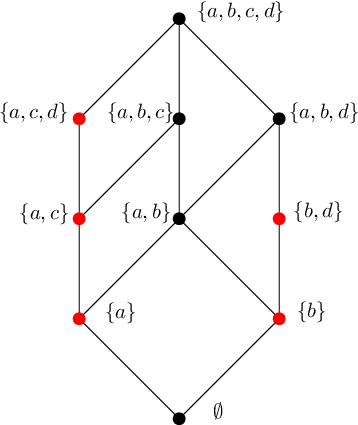}
        \caption{Antimatroid $(V, \mathcal G)$, partially ordered by inclusion.}
         \end{subfigure}
        \begin{subfigure}[t]{.4\linewidth}
        \centering\includegraphics[scale=0.5]{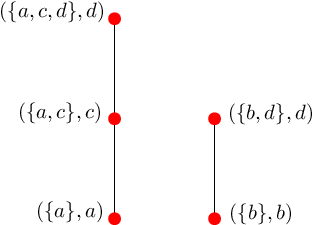}
        \caption{Path poset $({\mathcal Q}, \supseteq)$, with each path denoted together with its endpoint.}
        \end{subfigure}
        \caption{An antimatroid, and its path poset. The nodes corresponding to paths are highlighted in red.}
        \label{fig:antimatroid}
    \end{figure}
\end{example}

Because of Proposition~\ref{prop:feasible_sets_representation}, an antimatroid can be encoded by its path poset. This encoding can be used to formulate the following optimization problem. 

\begin{definition}[Minimum cost feasible set in an antimatroid]
    Let $(V, \mathcal G)$ be an antimatroid encoded by its path poset $({\mathcal Q}, \supseteq)$. Let $c: V\to \mathbb Z$. The goal of the \emph{minimum cost feasible set in $({\mathcal Q},\supseteq)$ problem} is to find an element in
    $\arg\min_{G\in \mathcal G} \sum_{g\in G}c(g)$.
\hfill $\diamond$ \end{definition}

\begin{theorem}[\citet{merckx2019optimization}]\label{thm:antimatroid_hardness}
    The minimum cost feasible set in an antimatroid problem is NP-hard.
\end{theorem}

For completeness, we present the proof of Theorem~\ref{thm:antimatroid_hardness} due to \citet{merckx2019optimization} in Appendix~\ref{sec:app:hardness}.

\subsection{Complement lattices and complement join constraints}

We present next some additional properties on posets and join constraints. In the following, fix a lattice $\mathcal L = (X, \succeq)$ and its \emph{complement lattice} $\mathcal L^c = (X, \preceq)$. Also fix a  partial representation $(B,\succeq_B)$ of ${\cal L}$ via the partial representation function $\psi$. Define an \emph{upper closed set} of $B$ to be $S\subseteq B$ such that $x\in S$ and $x'\succeq x$ implies $x'\in S$. The set of upper closed sets of $B$ is $\mathcal D^c(B)$. Note that a set is upper closed if and only if its complement is lower closed. Also define $\psi^c: X\to \mathcal D^c(B)$ as $\psi^c(x) = B\setminus \psi(x)$. We can show that a partial representation of a lattice is the complement of a partial representation for the complement lattice.

\begin{lemma}[Complement of partial representation is partial representation of complement lattice]\label{lem:complement_partial_rep}
     $(B, \preceq_B)$ is a partial representation for $\mathcal L^c$, with partial representation function $\psi^c: X\to \mathcal D^c(B)$.
\end{lemma}
\begin{proof}
    Recall that the complement of a lower closed set is an upper closed set. Thus, $\psi^c(x) = B\setminus \psi(x)$ is an upper closed set of $(B, \succeq_B)$, which is also a lower closed set of $(B, \preceq_B)$, since the direction of the partial order is reversed. 

    Next, we will show that $\psi^c$ is an order-embedding. Fix any $x, x'\in X$. Since $\psi$ is a partial representation function and thus an order-embedding, we know that $x\succeq x'$ if and only if $\psi(x)\supseteq \psi(x')$ if and only if $\psi^c(x)\subseteq \psi^c(x')$. Thus, $\psi^c$ is an order-embedding from the complement lattice $(X,\preceq)$ into $(\mathcal{D}^c(B), \supseteq)$.

    Finally, we know that the maximal element $\overline x$ of the complement lattice $(X, \preceq)$ is the minimal element of the original lattice $(X, \succeq)$. So, we have $\psi^c(\overline x) = B\setminus \psi(\overline x) = B\setminus\emptyset = B$. Similarly, the minimal element $\underline x$ of the complement lattice is the maximal element of the original lattice, so we have $\psi^c(\underline x) = B\setminus \psi(\underline x) = B\setminus B = \emptyset$. We conclude that $(B, \preceq_B)$ is a partial representation for $(X, \preceq)$ with partial representation function $\psi^c$.
\end{proof}

\begin{definition}[Complement of a join constraint]
    Given a join constraint $(\alpha,\beta)$ over $\mathcal L$ defined as
        $$\alpha(T) = (\mathbf{1}_{b_{11}}(T)\vee \mathbf{1}_{b_{12}}(T)\vee\dots)\wedge (\mathbf{1}_{b_{21}}(T)\vee \mathbf{1}_{b_{22}}(T)\vee\dots)\wedge\dots, \quad \beta(T) = \mathbf{1}_{b_{1}}(T)\wedge\mathbf{1}_{b_{2}}(T)\wedge\dots,$$
the \emph{complement} $(\beta^c,\alpha^c)$ of $(\alpha,\beta)$ is defined as $$\beta^c(T) = \mathbf{1}_{b_1}(T)\vee \mathbf{1}_{b_2}(T)\vee \cdots, \quad \alpha^c(T) = (\mathbf{1}_{b_{11}}(T)\wedge \mathbf{1}_{b_{12}}(T)\wedge\dots)\vee (\mathbf{1}_{b_{21}}(T)\wedge \mathbf{1}_{b_{22}}(T)\wedge\dots)\vee\dots.$$ 
\hfill $\diamond$ \end{definition}

Now, we show that the complement lattice satisfies the complement join constraint.

\begin{lemma}[Complement lattice satisfies complement join constraint]\label{lem:complement_lattice_contra_jc}
    Assume that the elements of $X$ satisfy a join constraint $(\alpha, \beta)$ on $B$ defined as in Definition~\ref{def:join_constraint}. Then $\psi^c$ is a partial representation function for $\mathcal L^c$. Furthermore, elements $x\in X$ satisfy $\beta^c(\psi^c(x)) = 1$ only if $\alpha^c(\psi^c(x)) = 1$.
\end{lemma}

\begin{proof}
    $\psi^c$ is a partial representation function for $\mathcal L^c$ by Lemma~\ref{lem:complement_partial_rep}.
    To show the second thesis, assume $\beta^c(\psi^c(x)) = 1$. Let $B_\beta$ be the elements of $B$ in the argument of $\beta$.
    Letting $\neg \chi = 1-\chi$, we write
    $$\begin{array}{rrl}       1 &= & \beta^c(B\setminus \psi(x)) = \bigvee_{b\in B_\beta} \mathbf{1}_{b}(B\setminus \psi(x))  = \bigvee_{b\in B_\beta} \neg\mathbf{1}_{b}(\psi(x)) = \neg\bigwedge_{b\in B_\beta}\mathbf{1}_{b}(\psi(x)).
    \end{array}$$
    Hence, we have
    $\beta(\psi(x))=  \bigwedge_{b\in B_\beta}\mathbf{1}_{b}(\psi(x)) 
        = 0.$
    Then, since $x\in X$ satisfies the join constraint $(\alpha, \beta)$, we know that $\alpha(\psi(x)) = 0$ as well.
    Thus,
    $$\begin{array}{ll}        1 &= \neg\alpha(\psi(x)) \\
        &= \neg ((\mathbf{1}_{b_{11}}(\psi(x))\vee \mathbf{1}_{b_{12}}(\psi(x))\vee\dots)\wedge (\mathbf{1}_{b_{21}}(\psi(x))\vee \mathbf{1}_{b_{22}}(\psi(x))\vee\dots)\wedge\dots) \\
        &= (\neg\mathbf{1}_{b_{11}}(\psi(x))\wedge \neg\mathbf{1}_{b_{12}}(\psi(x))\wedge\dots)\vee (\neg\mathbf{1}_{b_{21}}(\psi(x))\wedge \neg\mathbf{1}_{b_{22}}(\psi(x))\wedge\dots)\vee\dots \\
        &= (\mathbf{1}_{b_{11}}(\psi^c(x))\wedge \mathbf{1}_{b_{12}}(\psi^c(x))\wedge\dots)\vee (\mathbf{1}_{b_{21}}(\psi^c(x))\wedge \mathbf{1}_{b_{22}}(\psi^c(x))\wedge\dots)\vee\dots \\
        &= \alpha^c(\psi^c(x)),
    \end{array}$$
    as required.
\end{proof}

Let $T\subseteq B$. If $\beta^c(T) = 1$ only if $\alpha^c(T) = 1$, we say that $T$ \emph{satisfies the complement join constraint} $(\beta^c, \alpha^c)$. If $x\in\mathcal L$ satisfies $\beta^c(\psi(x)) = 1$ only if $\alpha^c(\psi(x)) = 1$, we say that $x$ satisfies $(\beta^c, \alpha^c)$ via $\psi$.

We next show that, analogously to Lemma~\ref{lem:poly_join_constrs}, we can use a polynomial number of complement join constraints to describe the path poset of an antimatroid, starting from the trivial poset over $V$, defined as $(V, \sim)$ where $x, y\in V$ satisfy $x\sim y$ if and only if $x=y$. Note that $(\mathcal D(V), \supseteq) = (2^V, \supseteq) = (\mathcal D^c(V), \supseteq)$, and $(\mathcal D(V), \supseteq)$ has trivial canonical representation function $\id: \mathcal D(V)\to \mathcal D(V)$.

\begin{lemma}[Polynomially many complement join constraints describe an antimatroid]\label{lem:poly_rjc_antimatroid}
    Let $(V, \mathcal G)$ be an antimatroid with path poset $({\mathcal Q}, \supseteq)$. 
    Then there is a set $\Omega$ of $O(|V|)$ complement join constraints on $(V, \sim)$ each with $O(|{\mathcal Q}|^2)$ arguments such that $$\mathcal G = \{T\in \mathcal D^c(V) \mid T\text{ satisfies } (\beta^c, \alpha^c)\ \forall (\beta^c, \alpha^c)\in \Omega\}.$$
\end{lemma}
\begin{proof}
    For each $x\in V$, define $\alpha_x, \beta_x: 2^V \to \{0, 1\}$ by
    $$\alpha_x(T) = \bigwedge_{G\in {\mathcal Q}(x)}\left(\bigvee_{G'\in \psi_{\mathcal G}(G)}\mathbf{1}_{e(G')}(T)\right), \quad 
    \beta_x(T) = \mathbf{1}_{x}(T),$$
    where $\psi_{\mathcal G}$ is defined as in Proposition~\ref{prop:feasible_sets_representation}. Note that for each $e(G'_1), e(G'_2)$ in the argument of $\alpha_x$, we have $e(G'_1)\not\prec e(G'_2)$, since $(V, \sim)$ is endowed with the trivial partial order. Thus, $(\alpha_x, \beta_x)$ is a join constraint. Moreover, $\alpha_x,\beta_x$ have $O(|{\mathcal Q}|^2)$ arguments. Then, define 
    $\Omega = \{(\beta_x^c, \alpha_x^c)\mid x\in V\}.$
    First, we show that $\mathcal G\subseteq  \{T\in \mathcal D^c(V) \mid T\text{ satisfies } (\beta^c, \alpha^c)\ \forall (\beta^c, \alpha^c)\in \Omega\}.$ It suffices to show that for any $G^*\in \mathcal G$ and $x\in V$, $G^*$ satisfies $(\beta_x^c, \alpha_x^c)$. So, suppose $\beta_x^c(G^*) = 1$. We can write this relation as $\beta_x^c(G^*) = \mathbf{1}_{x}(G^*) =1$, which means that $x\in G^*$. By Proposition~\ref{prop:paths_contain_all_endpoints}, there exists a path $\tilde G\subseteq G^*$ such that $e(\tilde G) = x$. We also know that for all $G'\in \psi_{\mathcal G}(\tilde G)$, we have $e(G')\in G'\subseteq \tilde G\subseteq G^*$ as well. Writing $\alpha^c_x(G^*)$ as
    $$\alpha^c_x(G^*) = \bigvee_{G\in {\mathcal Q}(x)}\left(\bigwedge_{G'\in \psi_{\mathcal G}(G)}\mathbf{1}_{e(G')}(G^*)\right),$$
    we see that the subclause of $\alpha^c_x(G^*)$ given by
    $\bigwedge_{G'\in \psi_{\mathcal G}(\tilde G)}\mathbf{1}_{e(G')}(G^*)$
    is equal to 1. Thus, $\alpha^c_x(G^*) = 1$, showing that $G^*$ satisfies all $(\beta^c, \alpha^c)\in \Omega$.

    To complete the proof, we show that $$\mathcal G\supseteq  \{T\in \mathcal D^c(V) \mid T\text{ satisfies } (\beta^c, \alpha^c)\ \forall (\beta^c, \alpha^c)\in \Omega\}.$$ Fix any $T\in \mathcal D^c(V)$ that satisfies all complement join constraints in $\Omega$. By Proposition~\ref{prop:feasible_sets_representation}, it suffices to show that $T$ is the union of paths in ${\mathcal Q}$. So, fix any $x\in T$. We claim that there exists $G\in {\mathcal Q}$ such that $x\in G\subseteq T$. Consider $(\beta_x^c, \alpha_x^c)\in\Omega$. We have $\beta_x^c(T) = \mathbf{1}_{x}(T) = 1$. Since $T$ must satisfy $(\beta^c, \alpha^c)$, we know that $\alpha_x^c(T) = 1$ as well. We can write this as
    $$\alpha_x^c(T) = \bigvee_{G\in {\mathcal Q}(x)}\left(\bigwedge_{G'\in \psi_{\mathcal G}(G)}\mathbf{1}_{e(G')}(T)\right) = 1.$$
    So, there is some $G\in {\mathcal Q}(x)$ such that 
    $\bigwedge_{G'\in \psi_{\mathcal G}(G)}\mathbf{1}_{e(G')}(T) = 1.$
    That is, $T$ contains the endpoints $e(G')$ of all paths $G'\subseteq G$. By Proposition~\ref{prop:paths_contain_all_endpoints}, every element in $G$ is the endpoint of some path $G'\subseteq G$. Thus, $G\subseteq T$. Iterating the argument we see that, for each $x \in T$, there exists a path $G$ with $x\in G\subseteq T$. 
    We conclude that $T$ is the union of the paths contained in $T$, and so $T\in \mathcal G$.
\end{proof}

\subsection{Rotations} As a last step before delving into the proof of Theorem~\ref{thm:main2}, we present some results on rotations associated to certain GI instances and one-to-one matching markets.

\begin{proposition}[\citet{gusfield1989stable}]\label{prop:rotations_disjoint}
    Let $(F, W, (\mathcal C_f)_{f\in F}, (\mathcal C_w)_{w\in W})$ be a one-to-one matching market instance with stable matching lattice $(\mathcal S, \succeq)$ and rotation poset $(\Pi, \succeq)$. Then, for all distinct rotations $\rho, \hat \rho \in \Pi$, $\rho^+\cap \hat \rho^+  = \emptyset$ and $\rho^- \cap \hat \rho^- = \emptyset$. Furthermore, if $(F, W, (\mathcal C_f)_{f\in F}, (\mathcal C_w)_{w\in W}) = I_{(X, \sim)}$ for a ground set $X$ with trivial partial order $\sim$, then we also have that for all $\rho, \hat \rho \in \Pi$, $\rho^+\cap \hat \rho^-  = \emptyset$.
\end{proposition}

Last, we show how to translate costs between rotations and matchings: see Appendix~\ref{sec:app:hardness} for a proof. 

\begin{lemma}[Cost function]\label{lem:cost_function}
    Let $I=I_{(V, \sim)}$ be the GI instance for $(V, \sim)$ with stable matching lattice $(\mathcal S, \succeq)$ and representation function $\psi_{\mathcal S}: \mathcal S\to \mathcal D(\Pi(\mathcal S))$.
    Let $\Omega$ be a set of join constraints over $(\mathcal S, \succeq)$, and let $I^* = (F^*, W^*, (\mathcal C^*_f)_{f\in F^*}, (\mathcal C^*_w)_{w\in W^*})$ 
    be an $\Omega$-extension of $I$ with stable matching lattice $(\mathcal S^*, \succeq)$, and partial representation function $\psi_{\mathcal S^*} = \psi_{\mathcal S}\circ \xi_{\mathcal S^*}$.
    Let $c: \Pi(\mathcal S)\to \mathbb{Z}$. Let $M = \operatorname{lcm}_{\rho\in \Pi(\mathcal S)} |\rho^-|$. Define $c': F^* \times W^* \rightarrow \mathbb{Z}$ as follows: for each $\rho\in \Pi$ and each $(f, w)\in \rho^-$, let $c'(f, w) = M \cdot c(\rho)/|\rho^-|$;
    for all other $(f, w)$, let $c'(f,w) = 0$. Then, $c'(\mu) = M \cdot c(\psi_{\mathcal S^*}^c(\mu))$ for all $\mu \in {\cal S}^*$.
\end{lemma}

In Lemma~\ref{lem:cost_function}, $c'$ is well-defined since for every pair $(f,w)$ of agents from opposite sides of the market of a one-to-one instance, $(f,w) \in \rho^-$ for at most one rotation $\rho$ (see Proposition~\ref{prop:rotations_disjoint}).

\section{Proof of Theorem~\ref{thm:main2}}\label{sec:proof:thm:main2}

We reduce from minimum cost feasible set in an antimatroid. Let $(V,{\cal G})$ be an antimatroid, $({\mathcal Q},\supseteq)$ its path poset, and $c: V \rightarrow \mathbb{Z}$ a cost function. %Let $(V,\sim)$ be the poset over $G$ such that $u \sim v$ if and only if $u = v$. 
Let $I=I_{(V, \sim)}$ be the GI instance for $(V, \sim)$ defined as in Theorem~\ref{thm:gusfield_irving} with stable matching lattice $(\mathcal S, \succeq)$, bijection $\phi: V\to \Pi(\mathcal S)$, and representation function $\psi_{\mathcal S}:\mathcal S\to \mathcal D(\Pi(\mathcal S))$. By Lemma~\ref{lem:poly_rjc_antimatroid}, there is a set $\Omega$ of $O(|V|)$ complement join constraints on $(V, \sim)$ each with $O(|{\mathcal Q}|^2)$ arguments such that $$\mathcal G = \{T\in \mathcal D^c(V) \mid T\text{ satisfies } (\beta^c, \alpha^c)\ \forall (\beta^c, \alpha^c)\in \Omega\}.$$

Let $\Omega' = \{(\alpha, \beta)\mid (\beta^c, \alpha^c)\in \Omega\}$ be the join constraints corresponding to the complement join constraints in $\Omega$. By Theorem~\ref{thm:omega_extension}, we can produce a standard $\Omega'$-extension $I^*$ of $I$, with matching lattice $(\mathcal S^*, \succeq)$ order-isomorphic to $\xi_{\mathcal S^*}(\mathcal S^*) = \{\mu\in\mathcal S\mid \mu\text{ satisfies }(\alpha,\beta),\ \forall (\alpha, \beta)\in\Omega'\},$ and partial representation function $\phi^{-1}\circ\psi_{\mathcal S^*}: \mathcal S^*\to \mathcal D(V)$, where $\psi_{\mathcal S^*} = \psi_{\mathcal S}\circ \xi_{\mathcal S^*}$. By Lemma~\ref{lem:complement_lattice_contra_jc}, we know that the complement lattice $(\mathcal S^*, \preceq)$ has partial representation $(\phi^{-1}\circ \psi_{\mathcal S^*})^c = \phi^{-1}\circ\psi_{\mathcal S^*}^c = V\setminus (\phi^{-1}\circ\psi_{\mathcal S^*}),$ and $\xi_{\mathcal S^*}(\mathcal S^*) = \{\mu\in \mathcal S\mid \mu \text{ satisfies }(\beta^c, \alpha^c) \text{ via }\phi^{-1}\circ \psi_{\mathcal S}^c,\ \forall (\beta^c, \alpha^c)\in\Omega\}.$ Thus, 
$$    \begin{array}{lll}
        \phi^{-1}(\psi_{\mathcal S^*}^c(\mathcal S^*)) & = &\{T \subseteq V \mid T = \phi^{-1}(\psi_{\mathcal S^*}^c(\mu)),\ \mu\in \mathcal S^*\} \\ & = &\{T\in \mathcal D^c(V) \mid T = \phi^{-1}(\psi_{\mathcal S^*}^c(\mu)),\ \mu\in \mathcal S^*\} \\
        &= &  \{T\in \mathcal D^c(V) \mid T = \phi^{-1}(\psi^c_{\mathcal S}(\mu)),\ \mu\in \mathcal S,\ \mu \text{ satisfies }\\ & & \quad(\beta^c, \alpha^c) \text{ via }\phi^{-1}\circ \psi_{\mathcal S}^c,\ \forall (\beta^c, \alpha^c)\in\Omega\}  \\
        &= &\{T\in \mathcal D^c(V) \mid T \text{ satisfies }(\beta^c, \alpha^c) \text{ via }\id,\ \forall (\beta^c, \alpha^c)\in\Omega\} \\ 
        &= & \mathcal G.
    \end{array}
 $$
    Extend $c$ to $\Pi(\mathcal S)$ by setting $c(\rho) = c(\phi^{-1}(\rho))$ for each $\rho\in \Pi(\mathcal S)$, and define $c'$ as in Lemma~\ref{lem:cost_function}. Specifically, let $M = \operatorname{lcm}_{\rho\in \Pi(\mathcal S)} |\rho^-|$. For each $\rho\in \Pi$ and each $(f, w)\in \rho^-$, we define the integer cost $c'(f, w) = M \cdot c(\rho)/|\rho^-|$. For all other $(f, w)$, set $c'(f,w) = 0$. Then we know that for each $\mu\in\mathcal S^*$ we have $c'(\mu) = M \cdot c(\psi_{\mathcal S^*}^c(\mu)) = M \cdot c(\phi^{-1}(\psi_{\mathcal S^*}^c(\mu)))$. Then, since $\phi^{-1}(\psi_{\mathcal S^*}^c(\mathcal S^*)) = \mathcal G$, the optimal value of $c'$ over $\mathcal S^*$ is exactly $M$ times the optimal value of $c$ over $\mathcal G$. Since $I^*$ can be obtained using $O(|V|)$ join constraint augmentations, each with $O(|{\mathcal Q}|^2)$ arguments, by Theorem~\ref{thm:omega_extension} we have a polynomial reduction from minimum cost feasible set in an antimatroid to minimum cost stable matching.
\hfill $\square$

\section*{Acknowledgements}

Supported by the grants NSF 2046146 and AFOSR FA9550-23-1-0697, and by a Cheung-Kong Innovation Doctoral Fellowship. The authors wish to thank Samuel Fiorini and a reviewer for suggesting papers~\cite{merckx2019optimization}and~\cite{fleiner2000stable}, respectively.

\newpage
\nocite{*}
\bibliographystyle{apalike}
\bibliography{main}

\newpage
\appendix

\section{Examples and discussions for Section \ref{sec:join_constraints}}\label{app:join_constraints}

We now start with the example. We then provide an intuitive description of the preference lists.

\begin{example}\label{ex:join_constr_aug}
    Consider the target non-distributive lattice $(X, \succeq)$ from Example~\ref{ex:distributive_closure}, as seen in Figure~\ref{fig:distributive_closure}(a), with poset of join-irreducible elements $(X_j, \succeq)$ in Figure~\ref{fig:distributive_closure}(b) and distributive closure $(\mathcal D(X_j), \supseteq)$ in Figure~\ref{fig:distributive_closure}(c). The corresponding GI instance $I=I_{(X_j, \succeq)}$ is given by firms $F = \{f_1,\cdots, f_7\}$, workers $W = \{w_1, \cdots, w_7\}$, and preference lists:
    \begin{align*}
        P_{f_1} &= (w_5, w_1) && P_{w_1} = (f_1, f_5) \\
        P_{f_2} &= (w_7, w_4, w_2) && P_{w_2} = (f_2, f_3) \\
        P_{f_3} &= (w_2, w_3) && P_{w_3} = (f_3, f_4, f_6) \\
        P_{f_4} &= (w_6, w_3, w_4) && P_{w_4} = (f_4, f_2, f_7) \\
        P_{f_5} &= (w_1, w_5) && P_{w_5} = (f_5, f_1) \\
        P_{f_6} &= (w_3, w_6) && P_{w_6} = (f_6, f_4) \\
        P_{f_7} &= (w_4, w_7) && P_{w_7} = (f_7, f_2)
    \end{align*}
    The stable matchings are then given by

        \medskip 

 \begin{center}   \boxed{
    \begin{aligned}
        {\color{blue}\mu_1 = \{(f_1, w_1), (f_2, w_2), (f_3, w_3), (f_4, w_4), (f_5, w_5), (f_6, w_6), (f_7, w_7)\}} \\
        {\color{blue}\mu_{10} = \{(f_1, w_5), (f_2, w_7), (f_3, w_2), (f_4, w_6), (f_5, w_1), (f_6, w_3), (f_7, w_4)\}}
    \end{aligned}
    }
    
    \boxed{
    \begin{aligned}
        {\color{red}\mu_2 = \{(f_1, w_5), (f_2, w_2), (f_3, w_3), (f_4, w_4), (f_5, w_1), (f_6, w_6), (f_7, w_7)\}} \\
        {\color{red}\mu_3 = \{(f_1, w_1), (f_2, w_4), (f_3, w_2), (f_4, w_3), (f_5, w_5), (f_6, w_6), (f_7, w_7)\}} \\
        {\color{red}\mu_5 = \{(f_1, w_1), (f_2, w_4), (f_3, w_2), (f_4, w_6), (f_5, w_5), (f_6, w_3), (f_7, w_7)\}} \\
        {\color{red}\mu_6 = \{(f_1, w_1), (f_2, w_7), (f_3, w_2), (f_4, w_3), (f_5, w_5), (f_6, w_6), (f_7, w_4)\}} \\
    \end{aligned}
    }
    
    \boxed{
    \begin{aligned}
        \mu_4 = \{(f_1, w_5), (f_2, w_4), (f_3, w_2), (f_4, w_3), (f_5, w_1), (f_6, w_6), (f_7, w_7)\} \\
        \mu_7 = \{(f_1, w_5), (f_2, w_4), (f_3, w_2), (f_4, w_6), (f_5, w_1), (f_6, w_3), (f_7, w_7)\} \\
        \mu_8 = \{(f_1, w_5), (f_2, w_7), (f_3, w_2), (f_4, w_3), (f_5, w_1), (f_6, w_6), (f_7, w_4)\} \\
        \mu_9 = \{(f_1, w_1), (f_2, w_7), (f_3, w_2), (f_4, w_6), (f_5, w_5), (f_6, w_3), (f_7, w_4)\} \\
    \end{aligned}
    }

\end{center}

    \medskip 
    
    The matchings in the first box (blue matchings) correspond to the minimal and maximal elements in the original lattice, as in Figure~\ref{fig:distributive_closure}. The matchings in the second box (red matchings) correspond to the join-irreducible elements of the original lattice. The remaining matchings (black matchings) are in the third box. The lattice of stable matchings $(S,\succeq)$ of $I$ is shown in Figure~\ref{fig:gusfield_irving_lattice_rotations}(a). We can compare it to the distributive closure $(\mathcal D(X_j), \supseteq)$ from Figure~\ref{fig:distributive_closure}(c) in Example~\ref{ex:distributive_closure}.
    \begin{figure}[H]
        \centering
        \begin{subfigure}[t]{.45\linewidth}
        \centering\includegraphics[scale=0.7]{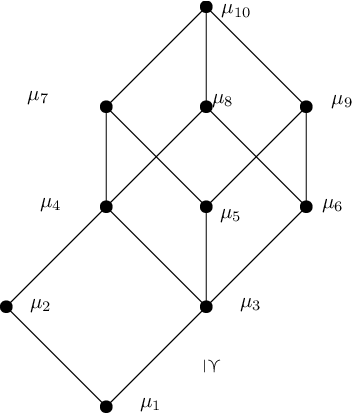}
        \caption[size=small]{Stable matching lattice $(\mathcal S,\succeq)$}
        \end{subfigure}
        \begin{subfigure}[t]{.45\linewidth}
        \centering\includegraphics[scale=0.7]{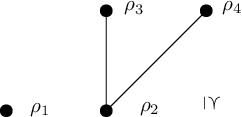}
        \caption[size=small]{Poset of join-irreducible elements $(X_j,\succeq)$, which is order-isomorphic to rotation poset $\Pi(\mathcal S)$}
        \end{subfigure}
        \caption[size=small]{The stable matching lattice, and poset of rotations.}
        \label{fig:gusfield_irving_lattice_rotations}
    \end{figure}
    By taking the minimal differences between stable matchings from $({\mathcal S},\succeq)$, we can find that the rotations are given by
    \begin{align*}
        \rho_1^+ &= \{(f_2, w_4), (f_3, w_2), (f_4, w_3)\} && \rho_1^- = \{(f_2, w_2), (f_3, w_3), (f_4, w_4)\} \\
        \rho_2^+ &= \{(f_1, w_5), (f_5, w_1)\} && \rho_2^- = \{(f_1, w_1), (f_5, w_5)\} \\
        \rho_3^+ &= \{(f_4, w_6), (f_6, w_3)\} && \rho_3^- = \{(f_4, w_3), (f_6, w_6)\} \\
        \rho_4^+ &= \{(f_2, w_7), (f_7, w_4)\} && \rho_4^- = \{(f_2, w_4), (f_7, w_7)\}
    \end{align*}
    The partial order over rotations is then shown in Figure~\ref{fig:gusfield_irving_lattice_rotations}(b). Again, it is order-isomorphic to the original partial order $(X_j,\succeq)$ as stated in Theorem~\ref{thm:gusfield_irving}. Next, our goal in this example is to augment the matching market instance so that the operation $\mu_2\vee \mu_3 = \mu_{10}$ holds in the augmented stable matching lattice. So, we can set 
    $$\alpha(T) = \mathbf{1}_{\rho_1}(T)\wedge\mathbf{1}_{\rho_2}(T),  
 \quad \beta(T) =  \mathbf{1}_{\rho_3}(T)\wedge\mathbf{1}_{\rho_4}(T) \notag.$$
    That is, the desired constraint is ``if $\rho_1$ and $\rho_2$ occur, then $\rho_3$ and $\rho_4$ occur.'' This ensures that if a stable matching $\mu$ satisfies $\mu\succeq\mu_2,\mu_3$, then $\psi_{\mathcal S}(\mu_2) = \{\rho_1\}$ and $\psi_{\mathcal S}(\mu_3) = \{\rho_2\}$ occur, and so $\{\rho_3, \rho_4\}$ occur in $\mu$ by the constraint. As a result, 
    $$\psi_{\mathcal S}({{\mu}}) \supseteq\{\rho_1, \rho_2, \rho_3, \rho_4\} \quad
    \implies \quad \mu \succeq \psi_{\mathcal S}^{-1}(\{\rho_1, \rho_2, \rho_3, \rho_4\}) = \mu_{10}.
    $$
    Applying Definition~\ref{def:join_constraint_augmentation}, we have $$\hbox{$\Pi_a=\{\rho_1,\rho_2\}$, $\Pi_b=\{\rho_3,\rho_4\}$, $F_{\alpha} = \{f_1, f_2, f_3, f_4, f_5\}$ and $W_{\beta} = \{w_3, w_4, w_6, w_7\}$.}$$ We add an auxiliary worker $w_0$ to $W_{aux}$ and an auxiliary firm $f_0$ to $F_{aux}$. Then, we add regular workers $W''_{\beta} = \{w_3'', w_4'', w_6'', w_7''\}$ and add them to the corresponding sets $copy(w_j)$. Given a set $T$ of firms, the choice function $\mathcal C''_{w_0}(T)$ for $w_0$ is the $(F_\alpha, f_0, \alpha(\pi(F_\alpha\setminus \cdot)))$-triggered choice function given by the process
    \begin{enumerate}[i.]
        \item Select $T\cap \{f_1, f_2, f_3, f_4, f_5\}$.
        \item If $\alpha(\pi(F_{\alpha}\setminus T)) = 1$, then additionally select $f_0$.
    \end{enumerate}
    Note that $\alpha(\pi(F_\alpha\setminus T))=1$ if and only if $T\cap F_\alpha=\emptyset$. Thus, ${\mathcal C}''_{w_0}(T)$ is equivalent to the preference list $P''_{w_0}$ given by
    $$P''_{w_0} = (\{f_1, f_2, f_3, f_4, f_5\}, \dots, \{f_0\})$$
    where in $(\dots)$ are all subsets of the set $F_{\alpha}=\{f_1, f_2, f_3, f_4, f_5\}$ (other than $F_{\alpha}$ itself) in order of decreasing length with ties broken arbitrarily. Given a set $T$ of workers, the new choice function $\mathcal C''_{f_0}(T)$ for $f_0$ is the if-else-$(w_0, W''_\beta)$ choice function given by the process
    \begin{enumerate}[i.]
        \item If $w_0\in T$, select $\{w_0\}$.
        \item Else, select $T\cap \{w_3'', w_4'', w_6'', w_7''\}$.
    \end{enumerate}
    This is equivalent to the preference list $P''_{f_0}$ given by
    $$P''_{f_0} = (\{w_0\}, \{w_3'', w_4'', w_6'', w_7''\}, \dots)$$
    where in $(\dots)$  are all subsets of the set $W_{\beta}''=\{w_3'', w_4'', w_6'', w_7''\}$ (other than $W''_{\beta}$ itself), in order of decreasing length with ties broken arbitrarily. To see the equivalence, notice that under $P''_{f_0}$, if $w_0\in T$, then $\{w_0\}$ is selected, as it is the most preferred set. Otherwise, the largest subset of $\{w''_3, w''_4, w''_6, w''_7\}$ is selected. 
    
    The preference lists for the workers in $W''_{\beta}$ are given by
    $$P''_{{w_3''}} = (\{f_0\}, \{f_6\}), \quad P''_{{w_4''}} = (\{f_0\}, \{f_7\}),\quad P''_{{w_6''}} = (\{f_0\}, \{f_4\}), \quad P''_{{w_7''}} = (\{f_0\}, \{f_2\}).$$
    For the firms $f\in F$, we first compute $$A''_{f_1} = \{(w_1, w_0)\},\qquad A''_{f_2} = \{(w_2, w_0)\},\qquad A''_{f_3} = \{(w_3, w_0)\},\qquad A''_{f_4} = \{(w_4, w_0)\}, \nonumber$$
    $$A''_{f_5} = \{(w_5, w_0)\},\qquad A''_{f_6} = \emptyset \qquad A''_{f_7} = \emptyset.$$
The choice functions can then be obtained following Definition~\ref{def:join_constraint_augmentation}. As an example, we consider $f=f_1,f_7$. Given a set $T$ of workers, the new choice function ${\mathcal C}''_{f_1}(T)$ for $f_1$ is the regular choice function given by the following process:
    \begin{enumerate}[i.]
        \item If $w_5\in T$, select $w_5$.
        \item Else, select $T\cap \{w_1,w_0\}$.
    \end{enumerate}
        This is equivalent to the preference list $P''_{f_1}$ given by
$$
P''_{f_1}=(\{w_5\},\{w_0,w_1\},\{w_1\},\{w_0\}).
$$
The new choice function ${\mathcal C}''_{f_7}(T)$ for $f_7$ is the regular choice function given by the following process:
    \begin{enumerate}[i.]
        \item If $\{w_4,{w_4''}\} \cap T\neq \emptyset$, select $\{w_4,{w_4''}\}\cap T$.
        \item Else, select $T\cap \{w_7,{w_7''}\}$.
    \end{enumerate}
        This is equivalent to the preference list $P''_{f_7}$ given by
$$
P_{f_7}''=(\{w_4,{w_4''}\},\{w_4\},\{{w_4''}\},\{w_7,{w_7''}\},\{w_7\},\{{w_7''}\}).
$$
        
    Finally, for each worker $w\in W$, the new choice functions $\mathcal C''_w$ are equal to the original choice functions $\mathcal C_w$.

    We can now examine the lattice of stable matchings $\mathcal S''$ in the instance post-augmentation. The stable matchings are given by
    \begin{gather*}
        \mu_1'' = \{(f_1, w_1), (f_1, w_0), (f_2, w_2), (f_2, w_0), (f_3, w_3), (f_3, w_0), (f_4, w_4), (f_4, w_0),\\ 
         \qquad (f_5, w_5), (f_5, w_0),  (f_6, w_6), (f_7, w_7), (f_0, {w_3''}), (f_0, {w_4''}), (f_0, {w_6''}), (f_0, {w_7''})\} \\
        \mu_2'' = \{(f_1, w_5), (f_2, w_2), (f_2, w_0), (f_3, w_3), (f_3, w_0), (f_4, w_4), (f_4, w_0), (f_5, w_1), \\
        \qquad (f_6, w_6), (f_7, w_7), (f_0, {w_3''}), (f_0, {w_4''}), (f_0, {w_6''}), (f_0, {w_7''})\} \\
        \mu_3'' = \{(f_1, w_1), (f_1, w_0), (f_2, w_4), (f_3, w_2), (f_4, w_3), (f_5, w_5), (f_5, w_0), (f_6, w_6),  \\
        \qquad (f_7, w_7), (f_0, {w_3''}), (f_0, {w_4''}), (f_0, {w_6''}), (f_0, {w_7''})\} \\
        \mu_5'' = \{(f_1, w_1), (f_1, w_0), (f_2, w_4), (f_3, w_2), (f_4, w_6), (f_5, w_5), (f_5, w_0), (f_6, w_3),  \\
        \qquad (f_7, w_7), (f_0, {w_3''}), (f_0, {w_4''}), (f_0, {w_6''}), (f_0,{ w_7''})\} \\
        \mu_6'' = \{(f_1, w_1), (f_1, w_0), (f_2, w_7), (f_3, w_2), (f_4, w_3), (f_5, w_5), (f_5, w_0), (f_6, w_6), \\
        \qquad (f_7, w_4), (f_0, {w_3''}), (f_0, {w_4''}), (f_0, {w_6''}), (f_0, {w_7''})\} \\
        \mu_9'' = \{(f_1, w_1), (f_1, w_0), (f_2, w_7), (f_3, w_2), (f_4, w_6), (f_5, w_5), (f_5, w_0), (f_6, w_3), \\
        \qquad (f_7, w_4), (f_0, {w_3''}), (f_0, {w_4''}), (f_0, {w_6''}), (f_0, {w_7''})\} \\
        \mu_{10}'' = \{(f_1, w_5), (f_2, w_7), (f_2, {w_7''}), (f_3, w_2), (f_4, w_6), (f_4, {w_6''}), (f_5, w_1), (f_6, w_3), \\
        \qquad (f_6, {w_3''}), (f_7, w_4), (f_7, {w_4''}), (f_0, w_0)\}
    \end{gather*}
    The stable matching lattice $\mathcal S''$ is then shown in Figure \ref{fig:join_constr_aug}. Notice that we can ``project'' each matching $\mu_i''$ back to the original instance by taking its intersection $\mu_i''\cap(F\times W)$ with the original firms and workers. Then, $\mu_i''\cap(F\times W) = \mu_i$, for $i \in\{ 1,2,3,5,6,9,10\}$. These are exactly the matchings in the original instance that satisfy the join constraint $(\alpha, \beta)$.

    In terms of the original target lattice $(X, \succeq)$, the join constraint ``if $\rho_1$ and $\rho_2$ occur, then $\rho_3$ and $\rho_4$ occur'' corresponds via $\phi^{-1}$ to the statement ``if $b$ and $c$ occur, then $d$ and $e$ occur.'' In other words, ``if $x\succeq b, c$, then $x\succeq d, e$.'' This is equivalent to exactly the join operation $b\vee c = f$ from the target lattice. Thus, using the join constraint and join constraint augmentation, we have enforced the equivalent of the join operation $b\vee c = f$ on the stable matching lattice. \hfill $\triangle$
\end{example}

\begin{figure}
    \centering
    \includegraphics[width=0.3\linewidth]{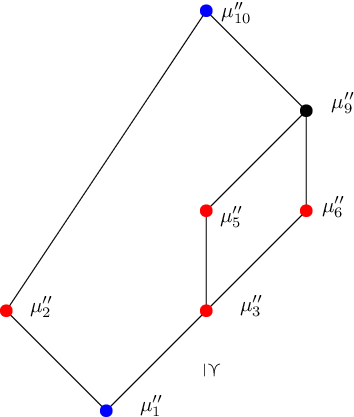}
    \caption{The stable matching lattice $\mathcal S''$ after applying the join constraint augmentation corresponding to $(\alpha,\beta)$.}
    \label{fig:join_constr_aug}
\end{figure}

\section{Proofs from Section~\ref{sec:hardness}}\label{sec:app:hardness}

For the sake of completeness, we reproduce here the proof of Theorem~\ref{thm:antimatroid_hardness}, which appeared in~\citet{merckx2019optimization}.

\begin{proof}[Proof of Theorem~\ref{thm:antimatroid_hardness}.] We reduce from the maximum independent set problem, which is well-known to be NP-Hard~\citep{karp2010reducibility}. Consider any graph $(V, E)$ for which we wish to find a maximum independent set. We construct an antimatroid $(V\cup E, \mathcal G)$ on the vertices and edges of the graph, such that $G\in \mathcal G$ is feasible if and only if, whenever $\{u, v\}\in G\cap E$, we have $u\in G$ or $v\in G$. Clearly, $V \cup E \in {\cal G}$ and ${\cal G}$ is closed under union. Moreover, for $G \in {\cal G}$, we can preserve feasibility by removing an edge if it contains one, or any element of $G$ if it does not. Hence, $(V\cup E,{\cal G})$ is an antimatroid. Define the weight function $w:V\cup E\to \mathbb Z$ by
    $$w(x) = \begin{cases}
        1 - d(x) & x\in V \\
        1 & x\in E
    \end{cases}$$
    where $d(v)$ is the degree of vertex $v$. Now, consider any independent set $U\subseteq V$. Let $G = U\cup \{\{u, v\}\in E\mid u\in U\text{ or }v\in U\}$. $G$ is thus a feasible set. Then by independence, $|G\cap E| = \sum_{u\in U}d(u)$. It follows that
    $$
    \begin{array}{lll}    w(G) &= \sum_{u\in U}w(u) + \sum_{e\in G\cap E}w(e) \\
        &= \sum_{u\in U}(1-d(u)) + \sum_{u\in U}d(u) \\
        &= |U|.
    \end{array}$$
    Thus, the maximum weight feasible set has weight at least equal to the maximum cardinality independent set. Now consider any feasible set in $G\in \mathcal G$. We construct an independent set of vertices $U$ from $G$. First, set $U\gets G\cap V$. If $U$ is independent, we are done. Otherwise, select any $\{u, v\}\in E$ such that $u,v\in G\cap V$ and remove one of its vertices, say $u$, from $G$. Additionally, remove any edge $\{u, v'\}\in G\cap E$ with $v'\not\in G$, to preserve feasibility of $G$. Note that this operations removes $u$ and at most $d(u)-1$ edges from $G$. Thus, the weight of $G$ does not decrease. Repeat this until the remaining vertices are independent. Let $G'$ be the final feasible set, and $U\gets G'\cap V$ be the independent set. Then, by definition of $w$, we can write
    $$\begin{array}{ll}    w(G) &\le w(G') \\
        &\le \sum_{v\in G'\cap V}(1 - d(v)) + \sum_{e\in G'\cap E}1 \\
        &\le \sum_{v\in G'\cap V}(1 - d(v)) + \sum_{v\in G'\cap V}d(v) \\
        &= |G'\cap V| \\
        &= |U|.
    \end{array}$$
    The maximum cardinality independent set is thus equal to the weight of the maximum weight feasible set. Additionally, the path poset of the antimatroid $(V\cup E, \mathcal G)$ has $O(|V|+|E|)$ elements, and thus has size polynomial in $|V|, |E|$. We conclude that computing the maximum weight feasible set in an antimatroid is NP-hard. By negating the weight function $w$, it immediately follows that finding the minimum cost feasible set in an antimatroid is also NP-hard, establishing Theorem~\ref{thm:antimatroid_hardness}.
\end{proof}

Next, we have the proof of Lemma~\ref{lem:cost_function}.

\begin{proof}[Proof of Lemma~\ref{lem:cost_function}.]
    Let $\mu \in {\cal S}^*$ and fix $R = \psi_{\mathcal S^*}^c(\mu)$. We can write
    \begin{equation}\label{eq:ecco-mu}\mu = (\psi_{\mathcal S^*}^c)^{-1}(R) = \psi_{\mathcal S^*}^{-1}(\Pi\setminus R)\end{equation}
    by the definitions of $R$, $\psi_{\mathcal S^*}^c$, and $\psi_{\mathcal S^*}$. Observe that \begin{equation}\label{eq:c'-and-mu}c'(\mu) = c'(\mu \cap (F\times W)) = c'(\xi_{\mathcal S^*}(\mu)),\end{equation} where the first equation holds since, by construction, only edges in $F\times W$ have nonzero weight under $c'$, and the second since, by Lemma~\ref{lem:from-mu-to-muprime-to-muprimeprime}, $\mu^*(w)=\xi_{\mathcal S^*}(\mu^*)(w)$ for all $w \in W$. Using~\eqref{eq:ecco-mu} and~\eqref{eq:c'-and-mu}, we can write
    $$c'(\mu)= c'(\xi_{\mathcal S^*}(\psi_{\mathcal S^*}^{-1}(\Pi\setminus R))) = c'(\psi_{\mathcal S}^{-1}(\Pi\setminus R)),$$
    as $\psi_{\mathcal S^*} = \psi_{\mathcal S}\circ\xi_{\cal S^*}$ and $\psi_{\mathcal S}^*$, $\xi_{\cal S^*}$ are injective (the first by definition of a partial representation, the second by Lemma~\ref{lem:join_constr_aug_order_embedding}). Then, using Theorem~\ref{thm:rotation_representation}, we can write
    \[ c'(\mu)= c'\left(\mu_W\cup\left(\bigcup_{\rho\in \Pi\setminus R}\rho^+\right)\setminus\left(\bigcup_{\rho\in \Pi\setminus R}\rho^-\right)\right),\label{thm:main2:mu_W}\]
    where $\mu_W$ is the minimal element of $(\mathcal S, \succeq)$. Next, note that again by Theorem~\ref{thm:rotation_representation}, we can write
    $$c'(\mu_F) = c'(\psi_{\mathcal S}^{-1}(\Pi)) = c'\left(\mu_W\cup\left(\bigcup_{\rho\in \Pi}\rho^+\right)\setminus\left(\bigcup_{\rho\in \Pi}\rho^-\right)\right) = 0,\label{thm:main2:muF_expanded}$$
    where $\mu_F$ is the maximal element of $(\mathcal S, \succeq)$, and the last equation holds by definition of $c'$. So, using \eqref{thm:main2:mu_W} we write
    \begin{equation}\label{eq:the-last-of-eq}c'(\mu)=c'\left(\mu_F\cup\left(\bigcup_{\rho\in R}\rho^-\right)\setminus\left(\bigcup_{\rho\in R}\rho^+\right)\right).\end{equation}
    Now, we know by the first thesis of Proposition~\ref{prop:rotations_disjoint} that each of the $\rho^-$ are disjoint from each other, and from its second thesis that $\bigcup_{\rho\in R}\rho^-$ and $\bigcup_{\rho\in R}\rho^+$ are disjoint, since the original GI instance is $I_{(V,\sim)}$. Since only $(f, w)\in \rho^-$ for some $\rho\in \Pi$ can have $c'(f, w)\ne 0$, from~\eqref{eq:the-last-of-eq} we deduce
    $$\begin{array}{lll} c'(\mu)& = & c'(\mu_F) + \sum_{\rho\in R}\sum_{(f, w)\in \rho^-}c'(f, w) \\ [1.2ex] & = &  \sum_{\rho\in R}\sum_{(f, w)\in \rho^-}\frac{M \cdot c(\rho)}{|\rho^-|} \\ [1.2ex]
    &= & \sum_{\rho\in R}M \cdot c(\rho) = M \cdot c(R) = M \cdot c(\psi_{\mathcal S^*}^c(\mu)),\end{array}$$
    as desired.
\end{proof}

\end{document}